\title{Index coding via linear programming}
\author{ {Anna Blasiak\thanks{Department of Computer Science, Cornell University, Ithaca NY 14853. E-mail: {\tt ablasiak@cs.cornell.edu}. Supported by an NDSEG Graduate Fellowship, an AT\&T
Labs Graduate Fellowship, and an NSF Graduate Fellowship.}} \qquad
{Robert Kleinberg\thanks{Department of Computer Science, Cornell University, Ithaca NY 14853. E-mail:
{\tt rdk@cs.cornell.edu}. Supported by NSF grant CCF-0729102, a grant from the Air Force Office of Scientific Research, a Microsoft Research New
Faculty Fellowship, and an Alfred P. Sloan Foundation Fellowship.}} \qquad
{Eyal Lubetzky\thanks{Microsoft
Research, One Microsoft Way, Redmond, WA 98052, USA. Email:
{\tt eyal@microsoft.com}.}}}
\documentclass [11pt]{article}
\usepackage[]{amsmath,amssymb,amsfonts,latexsym,amsthm,enumerate,paralist,url}
\usepackage{amssymb,graphicx}
\usepackage{booktabs}
\usepackage[numeric,initials,nobysame]{amsrefs}
\usepackage[compact]{titlesec}

\date{}

\numberwithin{equation}{section}

\newtheorem{maintheorem}{Theorem}

\newtheorem{theorem}{Theorem}[section]
\newtheorem*{theorem*}{Theorem}
\newtheorem{lemma}[theorem]{Lemma}
\newtheorem{claim}[theorem]{Claim}

\newtheorem*{observation*}{Observation}
\newtheorem{corollary}[theorem]{Corollary}

\newtheorem{fact}[theorem]{Fact}

\theoremstyle{definition}{

\newtheorem{definition}[theorem]{Definition}

\newtheorem{remark}[theorem]{Remark}
\newtheorem*{remark*}{Remark}
}

\newcommand{\deq}{\stackrel{\scriptscriptstyle\triangle}{=}}

\newcommand{\Z}{\mathbb Z}
\newcommand{\F}{\mathbb F}

\renewcommand{\P}{\mathbb{P}}

\newcommand{\cF}{\mathcal{F}}

\newcommand{\cC}{\mathcal{C}}

\renewcommand{\epsilon}{\varepsilon}

\newcommand{\chibar}{\overline{\chi}}
\newcommand{\minrk}{\operatorname{minrk}}
\newcommand{\rank}{\operatorname{rank}}
\newcommand{\conf}{\mathfrak{C}}
\newcommand{\orprod}{\ensuremath{\mathaccent\cdot\vee}}
\newcommand{\aac}{{\mathsf{AAC}}}

\newcommand{\hb}[1]{{\overline{H}({#1})}}

\newcommand{\hclq}{{\mathcal{J}}}
\newcommand{\decode}{\rightsquigarrow}
\newcommand{\Encode}{{\mathcal{E}}}
\newcommand{\Decode}{{\mathcal{D}}}
\newcommand{\msg}{V}
\newcommand{\weakhyp}{\psi_f}

\newcommand{\xhdr}[1]{\paragraph{#1.}}
\newcommand{\groundset}{V}
\newcommand{\fcn}{F}

\newif\ifshortver
\shortvertrue

\newcommand{\svapdx}[2]{\ifshortver#1\else#2\fi}
\newcommand{\svlabel}[1]{\ifshortver \label{sv:#1} \else \label{#1} \fi}
\newcommand{\svref}[1]{\ifshortver\ref{sv:#1}\else\ref{#1}\fi}
\newcommand{\sveqref}[1]{\ifshortver\eqref{sv:#1}\else\eqref{#1}\fi}

\begin{document}
\maketitle
\vspace*{-5mm}

\begin{abstract}
Index Coding has received considerable attention recently motivated in part by applications such as fast video-on-demand and efficient communication in wireless networks and in part by its connection to Network Coding.
Optimal encoding schemes and efficient heuristics were studied in various settings, while also leading to new results for Network Coding such as improved gaps between linear and non-linear capacity as well as hardness of approximation.
The basic setting of Index Coding encodes the side-information relation, the problem input, as an undirected graph and the fundamental parameter is the broadcast rate $\beta$, the average communication cost per bit for sufficiently long messages (i.e.\ the non-linear vector capacity).
Recent nontrivial bounds on $\beta$ were derived from the study of other Index Coding capacities (e.g.\ the scalar capacity $\beta_1$) by Bar-Yossef \emph{et al} (2006), Lubetzky and Stav (2007) and Alon \emph{et al} (2008).
However, these indirect bounds shed little light on the behavior of $\beta$: there was no known polynomial-time algorithm for approximating $\beta$ in a general network to within a nontrivial (i.e.\ $o(n)$) factor, and the exact value of $\beta$ remained unknown for \emph{any graph} where Index Coding is nontrivial.

Our main contribution is a direct information-theoretic analysis of the broadcast rate $\beta$ using linear programs, in contrast to previous approaches that compared $\beta$ with graph-theoretic parameters.
This allows us to resolve the aforementioned two open questions.  We provide a polynomial-time algorithm with a nontrivial approximation ratio for computing $\beta$ in a general network along with a polynomial-time decision procedure for recognizing instances with $\beta=2$.   In addition, we pinpoint $\beta$ precisely for various classes of graphs (e.g.\ for various Cayley graphs of cyclic groups) thereby simultaneously improving the previously known upper and lower bounds for these graphs.
Via this approach we construct graphs where the difference between $\beta$ and its trivial lower bound is linear in the number of vertices and ones where $\beta$ is uniformly bounded while its upper bound derived from the naive encoding scheme is polynomially worse.
\end{abstract}

\shortverfalse

\section{Introduction}\svlabel{sec:intro}

In the Index Coding problem a server holds a set of messages that it wishes to broadcast over a noiseless channel to a set of receivers. Each receiver is interested in one of the messages and has side-information comprising some subset of the other messages. Given the side-information map as an input, the objective is to devise an optimal encoding scheme for the messages (e.g., one minimizing the broadcast length) that allows all the receivers to retrieve their required information.

This notion of source coding that optimizes the encoding scheme given the side-information map of the clients was introduced by Birk and Kol~\cite{BK} and further developed by Bar-Yossef \emph{et al.} in~\cite{BBJK}. Motivating applications include satellite transmission of large files (e.g.\ video on demand), where a slow uplink may be used to inform the server of the side-information map, namely the identities of the files currently stored at each client due to past transmissions. The goal of the server is then to issue a shortest possible broadcast that allows every client to decode its target file while minimizing the overall latency. See~\cites{BK,BBJK,CS} and the references therein for further applications of the model and an account of various heuristic/rigorous Index Coding protocols.

The basic setting of the problem (see~\cite{AHLSW}) is formalized
as follows:
the server holds $n$ messages $x_1,\ldots,x_n \in \Sigma$ where
$|\Sigma|> 1$, and there are $m$ receivers
$R_1,\ldots,R_m$.  Receiver $R_j$ is interested in one message,
denoted by $x_{f(j)}$, and knows some subset $N(j)$ of the other
messages.
A solution of the problem must specify a finite alphabet $\Sigma_P$
to be used by the server, and an encoding scheme
$\Encode: \Sigma^n \to \Sigma_P$ such that,
for any possible values of $x_1,\ldots,x_n$,
every receiver $R_j$ is able to decode the message
$x_{f(j)}$ from the value of $\Encode(x_1,\ldots,x_n)$ together
with that receiver's side-information.  The minimum
encoding length $\ell = \left\lceil \log_2 |\Sigma_P|\right\rceil$ for
messages that are $t$ bits long (i.e.~$|\Sigma|=2^t$)
is denoted by $\beta_t(G)$, where $G$ refers to the data
specifying the communication requirements, i.e.~the functions
$f(j)$ and $N(j)$.
As noted in \cite{LuSt}, due to the  overhead associated
with relaying the side-information map to the server the
main focus is on the case $t\gg1$ and namely on the
following \emph{broadcast rate}.
\begin{equation}
  \svlabel{eq-beta-limit}
  \beta(G) \deq \lim_{t\to\infty}\frac{\beta_t(G)}t = \inf_t \frac{\beta_t(G)}{t}
\end{equation}
(The limit exists by sub-additivity.) This is interpreted as the average asymptotic number of broadcast bits needed per bit of input, that is, the asymptotic broadcast rate for long messages.
 In Network Coding terms, $\beta$ is the \emph{vector} capacity whereas $\beta_1$ is a \emph{scalar} capacity.

An important special case of the problem arises when there is exactly
one receiver for each message, i.e.~$m=n$ and $f(j)=j$
for all $j$.  In this case, the side-information map $N(j)$ can equivalently
be described in terms of the binary relation consisting of pairs
$(i,j)$ such that $x_j \in N(i)$.  These pairs can be thought of as
the edges of a directed graph on the vertex set $[n]$ or, in case the
relation is symmetric, as the edges of an undirected graph.
This special case of the problem (which we will hereafter identify
by stating that \emph{$G$ is a graph}) corresponds to the original
Index Coding problem introduced by Birk and Kol~\cite{BK}, and
has been extensively studied due to its rich connections with
graph theory and Ramsey theory.  These connections stem from
simple relations between broadcast rates and other graph-theoretic
parameters.
Letting $\alpha(G),\chibar(G)$ denote the independence and clique-cover numbers of $G$, respectively, one has
 \begin{equation}
   \svlabel{eq-trivial-ineqs}
   \alpha(G) \leq \beta(G) \leq \beta_1(G)\leq \chibar(G)\,.
 \end{equation}
The first inequality above is due to an independent set being identified with a set of receivers with no mutual information, whereas the last one due to~\cites{BK,BBJK} is obtained by broadcasting the bitwise XOR of the vertices per clique in the optimal clique-cover of $G$.

\svapdx{      \xhdr{History of the problem}}
       {\subsection{History of the problem}}\svlabel{subsec:related}
The framework of graph Index Coding and its scalar capacity $\beta_1$
were introduced in~\cite{BK}, where Reed-Solomon based protocols
hinging on a greedy clique-cover (related to the bound $\beta_1 \leq
\chibar$) were proposed and empirically analyzed.
In a breakthrough paper~\cite{BBJK}, Bar-Yossef \emph{et al.}\ proposed
a new class
of linear index codes based on a matrix rank minimization problem. The
solution to this problem, denoted by $\minrk_2(G)$, 
was shown to achieve the optimal linear scalar capacity over $GF(2)$ and in particular to be superior to the clique-cover method, i.e.\ $\beta_1 \leq \minrk_2 \leq \chibar$.
The parameter $\minrk_2$ was extended to general fields in~\cite{LuSt},
where arguments from Ramsey Theory
showed that for any $\epsilon>0$ there is a family of graphs on $n$ vertices where $\beta_1 \leq n^\epsilon$ while $\minrk_2 \geq n^{1-\epsilon}$ for any fixed $\epsilon>0$.
The first proof of a separation $\beta < \beta_1$ for graphs was presented by Alon \emph{et al.}\ in~\cite{AHLSW}; 
the proof introduces  a new capacity parameter $\beta^*$ such that $\beta \leq \beta^* \leq \beta_1$ and shows that the second inequality can be strict using a graph-theoretic characterization of  $\beta^*$.
In addition, the paper studied hypergraph Index Coding (i.e.~the general broadcasting with side information problem, as defined above), for which several hard instances were constructed --- ones where $\beta=2$ while $\beta^*$ is unbounded and others where $\beta^*<3$ while $\beta_1$ is unbounded.
The first proof of a separation $\alpha < \beta$ for graphs is presented in a companion paper~\cite{BKL11a}; the proof makes use of a new technique for bounding $\beta$ from below using a linear program whose constraints express information inequalities.  The paper then uses lexicographic products to amplify this separation, yielding a sequence of graphs in which the ratio $\beta/\alpha$ tends to infinity.  The same technique of combining linear programs with lexicographic products also leads to an unbounded multiplicative separation between non-linear and vector-linear Index Coding in hypergraphs.

As is clear from the foregoing discussion,
the prior work on Index Coding has been highly successful in
bounding the broadcast rate above and below by various parameters
(all of which are, unfortunately, NP-hard to compute) and in coming
up with examples that exhibit separations between these parameters.
However it has been less successful at providing general techniques
that allow the determination (or even the approximation)
of the broadcast rate $\beta$ for large classes of problem
instances.  The following two facts starkly illustrate this
limitation.  First, the exact
value of $\beta(G)$ remained unknown for \emph{every}
graph $G$ except those for which trivial lower and upper
bounds $\alpha(G), \chibar(G)$ coincide.  Second, it was not
known whether the broadcast rate $\beta$ could be approximated
by a polynomial-time algorithm whose approximation ratio
improves the trivial factor $n$ (achieved by simply broadcasting
all $n$ messages) by more than a constant
factor.\footnote{When $G$ is a graph, it is not hard to derive a
polynomial-time $o(n)$-approximation from~\sveqref{eq-trivial-ineqs}.}

In this paper, we extend and apply the linear programming
technique recently introduced in~\cite{BKL11a} to obtain
a number of new results on Index Coding, including
resolving both of the open questions stated in the preceding
paragraph.
The following two sections discuss our contributions, first to the
general problem of broadcasting with side information, and then to
the case when $G$ is a graph.

\medskip

\svapdx{      \xhdr{New techniques for bounding and approximating
the broadcast rate}}
       {\subsection{New techniques for bounding and approximating
the broadcast rate}}

\svlabel{subsec:new-techniques}

The technical tool at the heart of our paper is a pair of linear programs
whose values bound $\beta$ above and below.  The linear program
that supplies the lower bound was introduced in~\cite{BKL11a}
and discussed above;
the one that supplies the upper bound is strikingly similar,
and in fact the two linear programs fit into a hierarchy
defined by progressively strengthening the constraint set
(although the relevance of the middle levels of this hierarchy
to Index Coding, if any, is unclear).

\begin{maintheorem}\svlabel{thm-hierarchy}
Let $G$ be a broadcasting with side information problem,
having $n$ messages and $m$ receivers.
There is an explicit sequence of $n$ information-theoretic
linear programs, each one a relaxation of its successors, whose
respective solutions $b_1 \leq b_2 \leq \ldots  \leq b_n$ are such
that:
\begin{compactenum}[(i)]
\item \svlabel{item-b2-bn}
The broadcast rate $\beta$ satisfies $b_2 \leq \beta \leq b_n$,
and both of the inequalities can be strict.
\item \svlabel{item-b1-bn}
When $G$ is a graph, the extreme LP solutions $b_1$
and $b_n$ coincide with the independence number $\alpha(G)$ and
the fractional clique-cover number $\chibar_f(G)$
 respectively.
\end{compactenum}
\end{maintheorem}

As a first  application of this tool, we obtain the following
pair of algorithmic results.

\begin{maintheorem}\svlabel{thm-hypergraph-approx}
Let $G$ be a broadcasting with side information problem,
having $n$ messages and $m$ receivers.
Then there is a polynomial time algorithm which computes a
parameter $\tau=\tau(G)$ such that
$1 \leq \frac{\tau(G)}{\beta(G)} \leq O\big(n \frac {\log\log n}{\log n}\big)$.
There is also a polynomial time algorithm to decide
whether $\beta(G)=2$.
\end{maintheorem}
\svapdx{}{In fact, the $O \big( n \frac{\log \log n}{\log n} \big)$
approximation holds in greater generality for
the \emph{weighted} case, where different messages may
have different rates (in the motivating applications this can correspond e.g.\ to a server that holds files of varying size).
The generalization is explained in
Section~\ref{sec:weighted-hypergraph}.
}

\svapdx{      \xhdr{Consequences for graphs}}
       {\subsection{Consequences for graphs}}
\svlabel{sec:graph-consequences}

In
\svapdx{Appendix~\ref{sec:beta-of-graphs}}
       {Section~\ref{sec:beta-of-graphs}}
we demonstrate the use of
Theorem~\svref{thm-hierarchy} to derive the exact value of $\beta(G)$
for various families of graphs by analyzing the LP solution $b_2$.
As mentioned above, the exact value of $\beta(G)$ was previously
unknown for any graph except when the trivial lower and upper
bounds --- $\alpha(G)$ and $\chibar(G)$ --- coincide, as happens
for instance when $G$ is a perfect graph.
Using the stronger lower and upper bounds $b_2$ and $b_n$,
we obtain
the exact value of $\beta(G)$ for all cycles and cycle-complements:
$\beta(C_n) = n/2$ and
$\beta(\overline{C_n}) = n / \lfloor \frac{n}{2} \rfloor$.
In particular this settles the
Index Coding problem for the $5$-cycle investigated
in~\cites{BBJK,BKL11a,AHLSW}, closing the gap between
$b_2(C_5) = 2.5$ and $\beta^*(C_5) = 5 - \log_2 5 \approx 2.68$.
These results also provide simple constructions of networks with gaps between
vector and scalar Network Coding capacities.

We also use Theorem~\svref{thm-hierarchy} to prove separation between broadcast rates and other graph parameters.
Our results, summarized in Table~\svref{tab-comparison},
improve upon several of the best previously known
separations.
Prior to this work there were no known graphs
$G$ where $\beta_1(G) - \beta(G) \geq 1$. (For the more
general setting of broadcasting with side information,
multiplicative gaps that were logarithmic in the number
of messages were established in~\cite{AHLSW}.)
In fact, merely showing that the 5-cycle satisfies
$2 \leq \beta < \beta_1 = 3$ required the involved
analysis of an auxiliary capacity $\beta^*$, discussed earlier in
Section~\svref{subsec:related}.
With the help of our linear programming bounds
(Theorem~\svref{thm-hierarchy}) we supply
in Section~\svref{subsec:cor-of-thm-1}
a family of graphs on $n$ vertices where $\beta_1 - \beta$
is linear in $n$, namely $\beta = n/2$ whereas
$\beta_1=(1-\frac15\log_2 5-o(1))n\approx 0.54 n$.

\begin{table}[t]
\centering
\begin{tabular}{cccc}
\toprule%
Capacities & Best previous  & New separation &  Appears in\\
compared   & bounds in graphs      &  results  & Section\\
\midrule
$\beta-\alpha$ & $\Theta \left(n^{0.56}\right)$
& $\Theta(n)$ & \svref{subsec:cor-of-thm-1}\\
\midrule[0.25pt]
$\beta$ vs.\ $\chibar_f$
& $\begin{array}{c}
  \beta \leq n^{o(1)}\\
  \chibar_f \geq n^{1-o(1)}
\end{array}$
& $\begin{array}{c}
  \beta = 3\\
  \chibar_f =\Omega( n^{1/4})
\end{array}$ & \svapdx{\svref{sec:graph-stub}}{\ref{sec:sep-beta-alpha}}\\
\midrule[0.25pt]
$\beta_1-\beta$ & $\approx 0.32$ & $\Theta(n)$ & \svref{subsec:cor-of-thm-1}\\
$\beta_1/\beta$ & $\approx1.32$ & $1.5-o(1)$ & \svref{subsec:cor-of-thm-1}\\
\midrule[0.25pt]
$\beta^*-\beta$ & --- & $\Theta(n)$ & \svref{subsec:cor-of-thm-1}\\
\bottomrule
\end{tabular}
\caption{New separation results for Index Coding capacities in $n$-vertex graphs} \svlabel{tab-comparison}
\end{table}

We turn now to the relation between $\beta(G)$ and $\chibar_f(G)$,
the upper bound provided by our LP hierarchy.  As mentioned earlier,
Lubetzky and Stav~\cite{LuSt} supplied, for every $\varepsilon>0$, a
family of graphs on $n$ vertices satisfying
$\beta(G) \leq \beta_1(G) < n^\varepsilon$
while $\chibar_f(G) > n^{1-\varepsilon}$,
thus implying that $\chibar_f(G)$ is not
bounded above by any polynomial function
of $\beta(G)$.  We strengthen this result
by showing that $\chibar_f(G)$ is not bounded
above by \emph{any} function of $\beta(G)$.
To do so, we use
a class of projective Hadamard graphs due to Erd\H{o}s and R\'enyi
to prove the following theorem in
\svapdx{Section~\svref{sec:graph-stub}}{Section~\svref{sec:sep-beta-alpha}}.
\begin{maintheorem}\svlabel{thm-beta-chif-gap}
  There exists an explicit family of graphs $G$ on $n$ vertices such that $\beta(G) = 3$ whereas the Index Coding encoding schemes based on clique-covers cost at least $\chibar_f(G) = \Theta(n^{1/4})$ bits.
\end{maintheorem}

Recall the natural heuristic approach to Index Coding: greedily cover
the side-information graph $G$ by $r \geq \chibar(G)$ cliques and send
the XORs of messages per clique for an average communication cost of $r$.
 A similar protocol based on Reed-Solomon Erasure codes was proposed by~\cite{BK} and was empirically shown to be effective on large random graphs. Theorem~\svref{thm-beta-chif-gap} thus presents a hard instance for this protocol,
 namely graphs where $\beta=O(1)$ whereas $\chibar(G)$ is polynomially large.

\section{Linear programs bounding the broadcast rate}\svlabel{sec:hierarchy}

In this section we present 
linear programs that bound
the broadcast rate $\beta$ below and above,
using an information-theoretic analysis.
We demonstrate this technique
by determining $\beta(C_5)$ precisely;
later, in \svapdx{Appendix~\ref{sec:beta-of-graphs}}
{Section~\ref{sec:beta-of-graphs}},
we determine $\beta$ precisely for various infinite
families of graphs.

\subsection{The LP hierarchy}\svlabel{subsec:def-hierarchy}


Numerous results in Network Coding theory
bound the Network
Coding rate (e.g.,~\cites{AHJKL,DFZ1,HKL,HKNW,SYC})
by combining entropy inequalities of two types.  The first is
purely information-theoretic and holds for any set of random
variables; the second is derived from the graph structure.  An
important example of the second type of inequality, that we
refer to as ``decoding'', enforces the following: if a set
of edges $A$ cuts off a set of edges $B$ from all the sources, then
any information on edges in $B$ is determined by information on edges
in $A$.  We translate this idea to the setting of Index Coding in
order to develop stronger lower bounds for the broadcast rate.

\begin{definition}
Given a broadcasting with side information problem and subsets of
messages $A,B$, we say that $A$ \emph{decodes} $B$ (denoted $A
\decode B$) if $A \subseteq B$ and for every message $x \in B
\setminus A$ there is a receiver $R_j$ who is interested in $x$
and knows only messages in $A$ (i.e.\
$x_{f(j)} = x$ and $N(j) \subseteq A$).
\end{definition}
\begin{remark}
For graphs, $A \decode B$ if $A \subseteq B$ and for every $v\in
B\setminus A$ all the neighbors of $v$ are in $A$.
\end{remark}

If we consider the Index Coding problem on $G$ and a valid solution
$\Encode$, then the relation $A \decode B$ implies
$H(A,\Encode(x_1,\ldots,x_n)) \ge H(B,\Encode(x_1,\ldots,x_n))$,
since for each
message in $B \setminus A$ there is a receiver who must be able to
determine the message from only the messages in $A$ and
the public channel $\Encode(x_1,\ldots,x_n)$.
(Here and in what follows we denote by $H(X,Y)$
the joint entropy of the random variables $X,Y$.)
Combining these decoding inequalities with purely
information-theoretic inequalities, one can prove
lower bounds on the entropy of the public channel,
a process formalized by a linear program (that we
denote by $\mathcal{B}_2$) whose solution $b_2$
constitutes a lower bound on $\beta$.
(See~\cites{BKL11a,Yeung} for more on information-theoretic LPs.)
Interestingly, $\mathcal{B}_2$ fits into a hierarchy of $n$
increasing linear programs such that the last LP in the
hierarchy gives an \emph{upper} bound on $\beta$.

\begin{definition}
For a broadcasting with side information problem on a set $\msg$ of $n$ messages, the \emph{$\beta$-bounding LP hierarchy} is the sequence of LPs, denoted by $\mathcal{B}_1, \mathcal{B}_2, \mathcal{B}_3, \ldots , \mathcal{B}_{n}$ with solutions $b_1, b_2, \ldots, b_{n}$, given by:

\[
\begin{array}{llc}
\hline
\multicolumn{3}{c}{\mbox{\emph{$k$-th level of the LP hierarchy for the broadcast rate}}}\\
\cline{1-3}
  \mbox{minimize $X(\emptyset)$} \\
  \mbox{subject to:} \\
  \qquad X(\msg) \ge n & & \mbox{(\textit{initialize})}\\
  \qquad X(\emptyset) \ge 0 &  & \mbox{(\textit{non-negativity})}\\
  \qquad X(S) + |T \setminus S| \geq X(T) & \forall S \subseteq T \subseteq \msg &\mbox{(\textit{slope})}\\
  \qquad X(T) \ge X(S) & \forall S \subseteq T \subseteq \msg &\mbox{(\textit{monotonicity})} \\
  \qquad X(A) \ge X(B) &  \forall A,B \subseteq  \msg \,:\, A \rightsquigarrow B &\mbox{(\textit{decode})} \\
  \qquad \sum_{T \subseteq R} (-1)^{|R \setminus T|} X( T \cup Z) \le  0 &
  \!\!\!\begin{array}
    {l}
    \forall R \subseteq \msg \,:\, 2 \le |R| \le k \\
     \forall Z \subseteq \msg \,:\,  Z \cap R = \emptyset
  \end{array} & \mbox{(\textit{$|R|$-th order submodularity})}\\
  \hline
\end{array}
\]
\svlabel{def:bbLP}
\end{definition}
\begin{remark}
The above defined \emph{2-th order submodularity} inequalities are equivalent to the classical submodularity inequalities whereby $X(S) + X(T) \ge X(S \cap T) + X(S \cup T)$ for all $S,T$.
\end{remark}

Theorem~\svref{thm-hierarchy} traps $\beta$ in the solution sequence of
the above-defined hierarchy and characterizes its extreme values for
graphs.  The proofs of these results appear in
\svapdx{Appendix~\ref{sec:hierarchy-proof}}
      {Section~\ref{sec:hierarchy-proof}},
and in what follows we\svapdx{}{ first} outline the arguments therein and the
intuition behind them.

As mentioned above, the parameter $b_2$ is the entropy-based lower bound via
Shannon inequalities that is commonly used in the Network Coding
literature.  To see that indeed $\beta\geq b_2$ we
interpret a solution to the broadcasting problem as
a feasible primal solution to $\mathcal{B}_2$ via the
assignment $X(A) = H(A \cup \Encode(x_1,\ldots,x_n))$.
The proof that $\alpha(G) = b_1(G)$ for graphs is similarly based on
constructing a feasible primal solution to $\mathcal{B}_1$,
this time via the assignment
$X(A) = |A| + \max \{|I| \,:\, \mbox{$I$ is an independent set
  disjoint from $A$}\}$.  (The existence of this primal
solution justifies the inequality $b_1 \leq \alpha$; the
reverse inequality is an easy consequence of the decoding,
initialization, and slope constraints.)

To establish that $\beta(G) \leq b_n(G)$ when $G$ is a graph
we\svapdx{}{ will} show that $b_n(G) = \chibar_f(G)$, the fractional
clique-cover number of $G$,
while $\chibar_f(G)$ is an upper bound on
$\beta$.
For a general broadcasting network $G$
we\svapdx{}{ will} follow the same
approach via
an analog of $\chibar_f$ for hypergraphs.
It turns out that there are two natural generalizations
of cliques and clique-covers
in the context of broadcasting with side information.
\begin{definition} \svlabel{def:hyperclique}
A \emph{weak hyperclique} of a broadcasting problem
is a set of receivers $\hclq$ such that
for every pair of distinct elements $R_i,R_j \in \hclq$,
$f(i)$ belongs to $N(j)$.
A \emph{strong hyperclique} is
a subset of messages $T \subseteq V$ such that for any
receiver $R_j$ that desires $x_{f(j)} \in T$ we have that
$T \subseteq N(j) \cup \{f(j)\}$.

A \emph{weak fractional hyperclique-cover} is
a function that assigns a non-negative weight to each
weak hyperclique, such that for every receiver $R_j$,
the total weight assigned to weak hypercliques
containing $R_j$ is at least 1.  A \emph{strong
fractional hyperclique-cover} is defined the
same way, except that the weights are assigned
to strong hypercliques and the coverage requirement
is applied to receivers rather than messages.
In both cases, the \emph{size} of the hyperclique-cover
is defined to be the sum of all weights.
\end{definition}
Observe that if $T$ is any set of messages and $\hclq$ is the set
of all receivers desiring a message in $T$, then $T$ is a strong
hyperclique if and only if $\hclq$ is a weak hyperclique.
However, it is not the
case that every weak hyperclique can be obtained from a strong
hyperclique $T$ in this way.

Observe also that if $\hclq$ is a
weak hyperclique and each of the
messages $x_{f(j)}\, (R_j \in \hclq)$
is a single scalar value in some field,
then broadcasting the sum of those values
provides sufficient information for each
$R_j \in \hclq$ to decode $x_{f(j)}$.
This provides an indication (though not
a proof) that $\beta$ is bounded above by
the weak fractional hyperclique cover number.
The proof of Theorem~\svref{thm-hierarchy}(\svref{item-b2-bn})
in fact
identifies $b_n$ as being equal to the \emph{strong} fractional
hyperclique-cover number, which is obviously greater than or
equal to its weak counterpart.  The role of the
$n^{\mathrm{th}}$-order submodularity constraints
is that they force the function $F(S) \stackrel{\Delta}{=}
X(\overline{S}) - |\overline{S}|$ to be a
\emph{weighted coverage function}.  Using this representation
of $F$ it is not hard to extract a fractional set cover of
$V$, and the sets in this covering are shown to be
strong hypercliques using the decoding constraints.


Finally, \svapdx{the proof}{we will show}
that one can have $\beta > b_2$ \svapdx{uses}{using} a
construction based on the V\'amos matroid following the approach used in
\cite{DFZ2} to separate the corresponding Network Coding
parameters. As for showing that one can have $\beta < b_n$,
we\svapdx{}{ will} in fact show that one can have $\beta < b_3 \leq b_n$.

We believe that the other parameters $b_3,\ldots,b_{n-1}$ have no relation to $\beta$, e.g.\ as noted above we show that there is a broadcasting instance for which $\beta < b_3$ and thus $b_3$ is not a lower bound on $\beta$.

\svapdx{
}{
\subsection{Proof of Theorem~\svref{thm-hierarchy}}\svlabel{sec:hierarchy-proof}

In this section we prove Theorem~\svref{thm-hierarchy} via a series of claims.
The main inequalities involving the broadcast rate $\beta$ are shown in \S\svref{subsec:ap-hierarchy} whereas the constructions demonstrating that these inequalities can be strict appear in \S\svref{subsec:strict}.

\subsubsection{Bounding the broadcast rate via the LP hierarchy}\svlabel{subsec:ap-hierarchy}

We begin by familiarizing ourselves with the framework of the LP-hierarchy through proving
the following straightforward claim regarding the LP-solution $b_1$ and the graph independence number.

\begin{claim}\svlabel{clm-b1-eq-alpha}
If $G$ is a graph then the LP-solution $b_1$ satisfies $b_1(G) = \alpha(G)$.
\end{claim}
\begin{proof}
In order to show that $b_1(G) \ge \alpha(G)$, let $I$ be an
independent set of maximal size in $G$.  Now, $V\setminus I
\rightsquigarrow V$ implies that $X(V\setminus I) \ge X(V) \ge n$ is
true for any feasible solution.  Additionally, $X(V\setminus I) \le
X(\emptyset) + |V\setminus I|$.  Combining these together, we get
$X(\emptyset) \ge |V| - |V\setminus I| = |I| = \alpha(G)$.
To prove $b_1(G) \le \alpha(G)$ we present a feasible solution to the
primal attaining the value $\alpha(G)$,
\begin{equation}
  \svlabel{eq-alpha-feasible-sol}
  X(S) = |S| + \max \{|I| \,:\,
\mbox{$I$ is an independent set disjoint from $S$} \}\,,
\end{equation}

We verify that the solution is feasible by checking that it satisfies all the
constraints of $\mathcal{B}_1$. The fact that $X(V) = n$ implies the initialization constraint is satisfied.
To prove the slope constraint,
for $S \subseteq T \subseteq V$ let $I, J$ be maximum-cardinality
independent sets disjoint from $S,T$ respectively.  Note that $J$
itself is disjoint from $S$, implying $|J| \le |I|$.  Thus we have
\[
X(T) = |T| + |J| = |S| + |T \setminus S| + |J| \le
|S| + |T \setminus S| + |I| = X(S) + |T \setminus S|.
\]
Note also that $I \setminus T$ is an independent set disjoint
from $T$, hence it satisfies $|I \setminus T| \le |J|$.  Thus
\[
X(T) = |T| + |J| \ge |T| + |I \setminus T|  =
|T \cup I| \ge |S \cup I| = |S| + |I| = X(S),
\]
which verifies monotonicity.  Finally, to prove
decoding let $A,B$ be any vertex sets
such that $A \decode B$.
Consider $G \setminus A$,
the induced subgraph of $G$ on vertex set $V \setminus A$.
Every vertex of $B \setminus A$
is isolated in $G \setminus A$, and
consequently if $I$ is a maximum-cardinality
independent set disjoint from $B$,
then $I \cup (B \setminus A)$ is an independent
set in $G \setminus A$.  Therefore,
\begin{equation*}
X(A) \ge |A| + |I| + |B \setminus A|
= |B| + |I| = X(B)\,.\qedhere
\end{equation*}
\end{proof}

We next turn to showing that $b_2$ is a lower bound on the broadcast rate.
\begin{claim}\svlabel{clm-b2-leq-beta}
The LP-solution $b_2$ satisfies $b_2(G) \leq \beta(G)$.
\end{claim}
\begin{proof}

Let $G$ be a broadcasting with side information problem with $n$
messages $\msg$ and $m$ receivers.  Consider the
message $P = \Encode(x_1,\ldots, x_n)$ that we send on the public
channel to achieve $\beta$.  Denote by $H$ the entropy function normalized so that $H(x_i) = 1$ for all $i$.  This induces a function from the power set of $\msg \cup P$ to $\mathcal{R}$ where $H(S) = |S|$ for any subset of messages $S$ and $H(P) = \beta$.

Now, let $X(S) = H(S,P)$ for $S \subseteq \msg$.  We will
show that $X$ satisfies all the constraints of the LP $\mathcal{B}_2$, implying $X$ it is a feasible solution $\mathcal{B}_2$.

First, $X(\msg) \ge n$ since $H(\msg,P) = H(\msg)$ and our
normalization has $H(\msg) = n$.   Non-negativity holds because $H(P)
\ge 0$.   The $X(\cdot)$ values satisfy monotonicity and submodularity
because entropy does.  Slope is implied by
the fact that entropy is submodular (that is, $H(S,P) + H(T\setminus
S) \ge H(T,P)$) together with our normalization.  Finally,
decoding is satisfied because the coding solution is valid: each receiver $R_j$ can determine its sought information from
$N(j)$ and the public channel.

This solution gives $X(\emptyset) = H(P) = \beta$ and since the LP is
stated as  a minimization problem it implies that $\beta$ is an upper
bound on its solution $b_2$.
\end{proof}

Next we prove that $\beta \le b_n$.  We do this in three parts.
First, for every
instance $G$ of the
broadcasting with side information problem,
we define a parameter $\chibar_f(G)$ be the
minimum size of a strong fractional hyperclique-cover; this parameter
specializes to the fractional clique-cover number when $G$ is a graph.
Next we show that $\beta \le \chibar_f$, and finally we prove
that $\chibar_f = b_n$.






\begin{claim}\svlabel{clm-beta-leq-chibarf}
For any broadcasting problem with side information, $G$, we have
$\beta(G) \leq \chibar_f(G)$.
\end{claim}
\begin{proof}
Let $\cC$ be the set of strong hypercliques in $G = (\msg,E)$.  If
$\chibar_f \leq w$ then there is a finite collection of ordered pairs
$\{(S,x_S) \,:\, S \in \cC\}$ where the $x_S$'s are positive rational
numbers satisfying
\begin{align*}
\sum_{S \in \cC} x_S  = w \,,\quad\mbox { and }\quad
\sum_{S \in \cC: \, x \in S} x_S \geq 1\mbox{ for all $x \in \msg$}\,.
\end{align*}
Let $q$ be a positive integer such that each of the numbers $x_S \,
(S \in \cC)$ is an integer multiple of $1/q.$  Set $p = q w$, noting
that $p$ is also a positive integer.  Letting $y_S = q x_S$ for every
$S \in \cC$, we have:
\begin{align}
\sum_{S \in \cC} y_S = p\,,\quad\mbox{ and }\quad
\sum_{S \in \cC: \, x \in S} y_S & \geq q\mbox{ for all $x \in V$}\,.\svlabel{eq:chibarf-q}
\end{align}
Replacing each pair $(S,y_S)$ with $y_S$ copies of the pair $(S,1)$ if
necessary, we can assume that $y_S = 1$ for every $S$.
Similarly, replacing each $S$ by a proper subset if necessary, we can
assume that the inequality~\sveqref{eq:chibarf-q} is tight for every
$x$.
(Note that this step depends on the fact that the collection of
strong hypercliques, $\cC$, is closed under taking subsets.)
Altogether we have a sequence of sets $S_1,S_2,\ldots,S_p$, each of which
is a  strong hyperclique in $G$, such that every message occurs in
exactly $q$ of these sets.

From such a set system it is easy to construct an index code
where every message has $q$ bits (i.e.\ $\Sigma = \{0,1\}^q$)
and the broadcast utilizes $p$ bits (i.e.\ $\Sigma_P = \{0,1\}^p$).
Indeed, for each message $x \in V$ let
$j_1(x) < j_2(x) < \cdots < j_q(x)$ denote the indices such that
$x \in S_j$ for $j \in \{j_1(x),j_2(x),\ldots,j_q(x)\}$.  If the
bits of message $x$ are denoted by
$b_1(x),b_2(x),\ldots,b_q(x)$ then for each $1 \leq i \leq p$
the $i$-th bit of the index code is computed by taking the
sum (modulo 2) of all bits $b_k(z)$ such that $z \in S_i$ and
$i=j_k(z)$.  Receiver $R = (S,x)$ is able to decode the $k^{\mathrm{th}}$
bit of $x$ by taking the $j_k(x)$-th bit of the index code and
subtracting various bits belonging to other messages $x' \in
S_{j_k(x)}$.  All of these bits are known to $R$ since 
$S_{j_k(x)}$ is a strong hyperclique containing $x$.
This confirms that
$\beta(G) \leq p/q = w$, as desired.
\end{proof}

It remains to characterize the extreme upper LP solution:
\begin{claim}\svlabel{clm-bn-eq-chibarf}
  The LP-solution $b_n$ satisfies $b_n(G) = \overline{\chi}_f(G)$.
\end{claim}

\begin{proof}
The proof hinges on the fact that the entire set of constraints of $\mathcal{B}_n$ gives a useful structural characterization of any feasible solution $X$. Once we have this structure it will be simple to infer the required result.

\begin{lemma}\svlabel{lem:coverage-functions}
A vector $X$ satisfies the slope constraint and the i-th order
submodularity constraints for $i \in \{2,\ldots,n\}$ if and only if
there exists a vector of non-negative numbers $w(T)$, defined for
every non-empty set of messages $T$, such that $X(S) = |S| + \sum_{T: T \not \subseteq S} w(T)$ for all $S \subseteq \msg$.
\end{lemma}

The proof of this fact is similar to a characterization of a weighted coverage function. While much of the proof is likely folklore, we include it in Section~\svref{sec:ap-coverage} for completeness.

Given this fact we now prove that $b_n(G) \ge \chibar_f(G)$ by showing that any solution $X$ having the form stated in Lemma \svref{lem:coverage-functions} is a fractional coloring of $\overline{G}$.  Thus, for the remainder of this
subsection, $X$ refers to a solution of $\mathcal{B}_n$ having value $b_n(G)$
and $w$ refers to the associated vector of non-negative numbers whose
existence is guaranteed by Lemma~\svref{lem:coverage-functions}.
\begin{fact}\svlabel{fact:sumT}
For every message $x \in \msg$, $\sum_{T \ni x } w(T) = 1$.
\end{fact}
To see this, observe that monotonicity and decoding imply that $X(\msg \setminus \{x\}) = X(\msg)$.  Lemma~\svref{lem:coverage-functions} implies that the right-hand-side is $n$ while the left-hand-side is $n-1+\sum_{T \ni x} w(T)$.
\begin{fact}\svlabel{fact:sumneigh}
For every receiver $R_j$, if $x$ denotes $x_{f(j)}$,
then
$\sum_{T : \, x \, \in \, T \, \subseteq \, N(j) \cup \{x\}} w(T) = 1$.
\end{fact}
Indeed, monotonicity and decoding imply that $X(N(j) \cup \{x\}) = X(N(j))$. Lemma~\svref{lem:coverage-functions} implies that the right side and left side differ by
$1-\sum_{T : \, x \, \in \, T \, \subseteq \, N(j) \cup \{x\}} w(T).$


For a message $x$, let $N(x) = \bigcap_{j:x =x_{f(j)}} N(j)$ be the
intersection of the side information for every receiver who wants to
know $x$.  By combining Facts~\svref{fact:sumT} and~\svref{fact:sumneigh} we find
that if $w(T)$ is positive then $T$ is contained in $N(x) \cup \{x\}$
for every $x$ in $T$.  Thus, we can infer the following:
\begin{corollary}\svlabel{cor:clique}
If $w(T) > 0$ then the set of receivers desiring messages in $T$ is a
strong hyperclique.
\end{corollary}

Now, to prove $b_n(G) \le \overline{\chi}_f(G)$ we show that if a
vector $w$ gives a feasible fractional coloring then $X(S) = |S| + \sum_{T: T \not \subseteq S} w(T)$ is feasible for the LP
$\mathcal{B}_n$.  By the argument made in the proof of Claim \svref{clm-beta-leq-chibarf} we can assume without
loss of generality that $\sum_{T \ni u} w(T) = 1 \; \forall u \in
V$.  $X$ has value equal to the fractional coloring because
$X(\emptyset) = \sum_{T} w(T)$.  Further,
Lemma~\svref{lem:coverage-functions} implies that $X$ satisfies the
$i$-th order submodularity constraints and slope.   It trivially
satisfies initialization and non-negativity.  To show that $X$
satisfies monotonicity it is sufficient to prove that $X(S \cup \{u\})
\ge X(S)$ for all $S \subseteq V, u \in V \setminus S$.  By definition, we have $X(S \cup \{u\}) - X(S) = 1 -
\sum_{T : \, u \, \in \, T \, \not \subseteq \, S} w(T)
$.  Additionally, we know $
\sum_{T : \, u \, \in \, T \, \not \subseteq \, S} w(T)
\le  \sum_{T: u \in T}
w(T) = 1$, where the last equality is because $w$ is a fractional
coloring.  Finally, for the decoding constraints, it is sufficient
to show that $X(A)\ge X(A\cup\{x\})$ for $A = N(j)$
where $R_j$ is a receiver who desires $x$.  By definition of $X$,
$X(A) - X(A\cup\{x\}) =
\sum_{T: \, x \, \in \,  T \, \subseteq \, N(j) \cup \{x\}} w(T)
- 1$.  Also, $
\sum_{T : \, x \, \in \, T \, \subseteq \, N(j) \cup \{x\}} w(T)
= \sum_{T \ni x} w(T) = 1$ because $T$
with $w(T) > 0 $ is a strong hyperclique.
\end{proof}

\subsubsection{Strict lower and upper bounds for the broadcast rate}\svlabel{subsec:strict}

\begin{claim}\svlabel{claim-beta-le-b3}
There exists a broadcasting with side information instance $G$ for which $\beta(G) < b_3(G)$.
\end{claim}
\begin{proof}
The construction is an extremely simple instance with only three
messages $\{a,b,c\}$ and three receivers $(\{a\}, b), (\{b\}, c),$
and $(\{c\}, a)$.  It is easy to see that $a \oplus b, b
\oplus c$ is a valid solution, and thus $\beta \le 2$. However,
using the 3rd-order submodularity constraint we have that
\[ X(ab) + X(bc) + X(ac) + X(\emptyset) \ge X(abc) + X(a) + X(b) +
X(c).\]
Combining that with decoding inequalities
\[X(a) \ge X(ab)\,,\quad X(b) \ge X(bc)\,,\quad X(c) \ge X(ac)\,,\]
together with the initialization inequality $X(abc) \ge 3$ now gives us that $b_3 = X(\emptyset) \ge 3$.
\end{proof}
}

\subsection{The broadcast rate of the 5-cycle}
As stated in Theorem~\svref{thm-hierarchy}, whenever the LP-solution $b_2$ equals $\chibar_f$ we obtain that $\beta$ is precisely this value, hence one may compute the broadcast rate (previously unknown for any graph) via a chain of entropy-inequalities.
We will demonstrate this in Section~\svref{sec:beta-of-graphs} by determining $\beta$ for several families of graphs, in particular for cycles and their complements (Theorem~\svref{thm-cycles}).  These seemingly simple cases were previously studied in~\cites{AHLSW,BBJK} yet their $\beta$ values were unknown before this work.

\begin{figure}[tb]
\begin{center}
\includegraphics[width=6.5in]{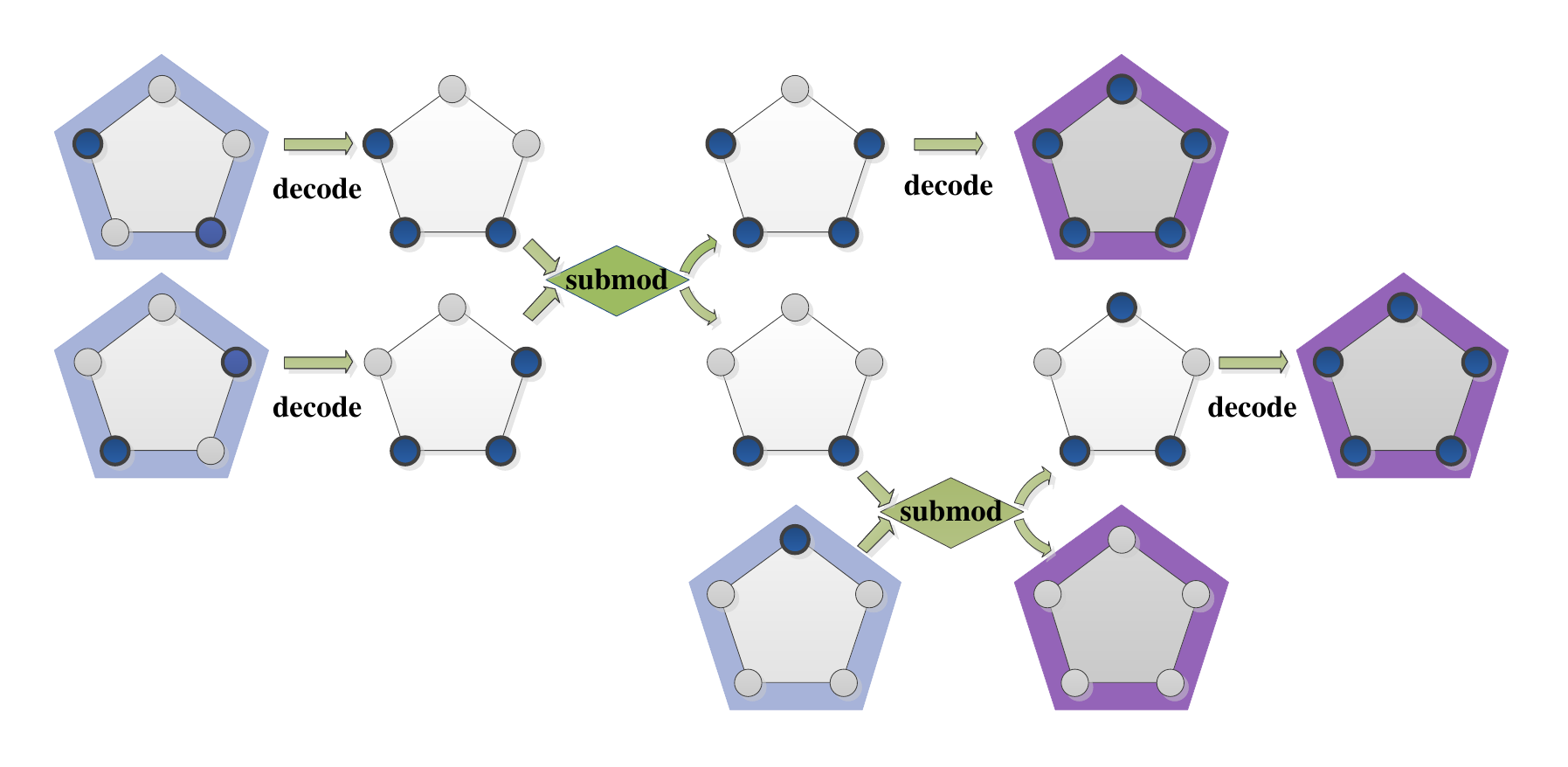}
\caption{A proof-by-picture that $\beta(C_5) = \frac52$. Variables marked by highlighted subsets of vertices, e.g.\ the first submodularity application applies the LP constraint $X(\{3,4,5\})+X(\{2,3,4\}) \geq X(\{2,3,4,5\})+X(\{3,4\})$. Final outcome is a proof that $\beta(C_5)\geq X(\emptyset)$ with $3X(\emptyset)+5 \geq X(\emptyset)+10$.}
\svlabel{fig:5cycle}
\end{center}
\vspace{-0.5cm}
\end{figure}
To give a flavor of the proof of Theorem~\svref{thm-cycles}, we provide a proof-by-picture for the broadcast rate of the 5-cycle (Figure~\svref{fig:5cycle}), illustrating the intuition behind choosing the set of inequalities one may combine for an analytic lower bound on $\beta$. The inequalities in Figure~\svref{fig:5cycle} establish that $\beta(C_5)\geq\frac52$, thus matching the upper bound $\beta(C_5) \leq \chibar_f(C_5) = \frac52$.

We note that odd cycles on $n\geq 5$ vertices as well as their complements constitute the first examples for graphs where the independence number $\alpha$ is strictly smaller than $\beta$. Corollary~\svref{cor-c5-union} will further amplify the gap between these parameters.

\subsection{Corollaries for vector/scalar index codes}\svlabel{subsec:cor-of-thm-1}

Prior to this work and its companion paper~\cite{BKL11a} there was no known family of graphs where $\alpha \neq \beta$,
and one could conjecture that for long enough messages the broadcast rate in fact converges to the independence number, the largest set of receivers that are pairwise oblivious. We now have that the 5-cycle provides an example where $\alpha = 2$ while $\beta =\frac52$, however here the difference $\beta-\alpha<1$ could potentially be attributed to integer-rounding, e.g.\ it could be that $\alpha = \lfloor \beta \rfloor$.

Such was also the case for the best known difference between the vector capacity $\beta$ and the scalar capacity $\beta_1$. The best lower bound on $\beta_1-\beta$ in any graph was again attained by the 5-cycle where it was slightly less than $\frac13$, and again in the constrained setting of graph Index Coding we could conjecture that $\beta_1 = \lceil \beta \rceil$.

The following corollary of the above mentioned results refutes these suggestions by amplifying both these gaps to be linear in $n$. The separation between $\alpha$ and $\beta$ was further strengthened in the companion paper~\cite{BKL11a}, where we obtained a gap of a polynomial factor between these parameters.
\begin{corollary}\svlabel{cor-c5-union}
There exists a family of graphs $G$ on $n$ vertices for which $\beta(G) = n/2$ while $\alpha(G) = \frac25 n$ and $\beta_1(G)=(1-\frac15\log_2 5+o(1))n\approx 0.54 n$. Moreover, we have $\beta^*(G) = (1-o(1))\beta_1(G)$.
\end{corollary}
To prove this result we will use the direct-sum capacity $\beta^*$. Recall that this capacity is defined to be $\beta^*(G)=\lim_{t\to\infty}\frac1{t}\beta_1(t\cdot G)=\inf_t \frac1{t}\beta_1(t\cdot G)$ where $t\cdot G$ denotes the disjoint union of $t$ copies of $G$. This parameter satisfies $\beta\leq\beta^*\leq\beta_1$.
Similarly we let $G+H$ denote the disjoint union of the graphs $G,H$. We need the following simple lemma.
\begin{lemma}\svlabel{lem-additive-beta-beta*}
The parameters $\beta$ and $\beta^*$ are additive with respect to disjoint unions, that is for any two graphs $G,H$ we have $\beta(G+H)=\beta(G)+\beta(H)$ and $\beta^*(G+H)=\beta^*(G)+\beta^*(H)$.
\end{lemma}
\begin{proof}[Proof of lemma]
The fact that $\beta^*$ is additive w.r.t.\ disjoint unions follows immediately from the results of~\cite{AHLSW}. Indeed, it was shown there that for any graph $G$ on $n$ vertices $\beta^*(G) = \log_2 \chi_f ( \conf(G) )$ where $\conf=\conf(G)$ is an appropriate undirected Cayley graph on the group $\Z_2^n$. Furthermore, it was shown that $\conf(G+H) = \conf(G) \orprod \conf(H)$, where $\orprod$ denotes the OR-graph-product. It is well-known (see, e.g.,~\cites{Feige,LV}) that the fractional chromatic number is multiplicative w.r.t.\ this product, i.e.\ $\chi_f(G \orprod H) = \chi_f(G)\chi_f(H)$ for any two graphs $G,H$. Combining these statements we deduce that
\begin{align*}
2^{\beta^*(G+H)} &= \chi_f( \conf(G+H) ) = \chi_f( \conf(G) \orprod \conf(H) ) = \chi_f(\conf(G))\chi_f(\conf(H)) = 2^{\beta^*(G)+\beta^*(H)}\,.
\end{align*}

We shall now use this fact to show that $\beta$ is additive. The inequality $\beta(G+H)\leq \beta(G)+\beta(H)$ follows from concatenating the codes for $G$ and $H$ and it remains to show a matching upper bound.

As observed by~\cite{LuSt}, the Index Coding problem for an $n$-vertex graph $G$ with messages that are $t$ bits long has an equivalent formulation as a problem on a graph with $t n$ vertices and messages that are $1$-bit long; denote this graph by $G_t$ (formally this is the $t$-blow-up of $G$ with independent sets, i.e.\ the graph on the vertex set $V(G) \times [t]$, where $(u,i)$ and $(v,j)$
are adjacent iff $u v \in E(G)$).
Under this notation $\beta_t(G) = \beta_1(G_t)$. Notice that $(G+H)_t = G_t+H_t$ for any $t$ and furthermore that $s\cdot G_t$ is a spanning subgraph of $G_{s t}$ for any $s$ and $t$, in particular implying that $\beta_1(s \cdot G_t) \geq \beta_1(G_{s t})$.

Fix $\epsilon > 0$ and let $t$ be a large enough integer such that $\beta(G+H) \geq \beta_t(G+H)/t -\epsilon$. Further choose some large $s$ such that $\beta^*(G_t) \geq \beta_1(s\cdot G_t)/s - \epsilon$ and
$ \beta^*(H_t) \geq \beta_1(s\cdot H_t)/s - \epsilon$.
We now get
\begin{align*}
\beta(G+H) + \epsilon &\geq \beta_1(G_t+H_t)/t \geq \beta^*(G_t+H_t)/t
 = \beta^*(G_t)/t+\beta^*(H_t)/t\,,
 \end{align*}
where the last inequality used the additivity of $\beta^*$. Since
\[ \beta^*(G_t)/t \geq \beta_1(s \cdot G_t)/st - \epsilon \geq \beta_1(G_{st})/st - \epsilon \geq \beta(G) -\epsilon\]
and an analogous statement holds for $\beta^*(H_t)/t$, altogether we have $\beta(G+H) \geq \beta(G) + \beta(H) - 3\epsilon$. Taking $\epsilon\to 0$ completes the proof of the lemma.
\end{proof}

\begin{proof}[\emph{\textbf{Proof of Corollary~\svref{cor-c5-union}}}]
Consider the family of graphs on $n=5k$ vertices given by $G = k \cdot C_5$. It was shown in~\cite{AHLSW} that $\beta^*(C_5) = 5-\log_2 5 $, which by definition implies that
$\beta^*(G) = (5-\log_2 5)k$ and $\beta_1(G) = \beta^*(G) + o(k)$.
At the same time, clearly $\alpha(G) = 2k$ and combining the fact that $\beta(C_5)=\frac52$ with Lemma~\svref{lem-additive-beta-beta*} gives $\beta(G) = 5k/2 = n/2$, as required.
\end{proof}

The above result showed that the difference between the broadcast rate $\beta$ and the Index Coding scalar capacity $\beta_1$ can be linear in the number of messages. We now wish to use the gap between $\beta$ and $\beta_1$ to infer a gap between the vector and scalar Network Coding capacities.

\begin{corollary}\svlabel{cor-nc-gap}
 For any $k\geq 1$ there exists a Network Coding instance on $5k+2$ vertices where the ratio between the vector and scalar-linear capacities is precisely $1.2$ while the ratio between the vector and scalar capacities converges to $1-\frac12\log_25 \approx 1.07$ as $k\to\infty$.
\end{corollary}
\begin{proof}
It is well known (e.g.~\cite{RSG}) that an $n$-vertex graph Index Coding instance $G$ can be translated into a capacitated network $H$ on $2n+2$ vertices via a reduction that preserves linear encoding. It thus suffices to bound the ratio of the corresponding Index Coding capacities.

For $k\geq 1$ consider the graph $G$ consisting of $k$ disjoint $5$-cycles. Corollary~\svref{cor-c5-union} established that $\beta(G) = 5k/2$ whereas $\beta_1(G) = (5-\log_2 5+o(1))k$ where the $o(1)$-term tends to $0$ as $k\to\infty$. At the same time, it was shown in~\cite{BBJK} that the scalar-linear Index Coding capacity over $GF(2)$ coincides with a parameter denoted by $\minrk_2(G)$, and as observed in~\cite{LuSt} this extends to any finite field $\F$ as follows: For a graph $H=(V,E)$ we say that a matrix $B$ indexed by $V$ over $\F$ is a \emph{representation} of $H$ over $\F$ if it has nonzero diagonal entries ($B_{uu}\neq 0$ for all $u\in V$) whereas $B_{u v} = 0$ for any $u\neq v$ such that $u v \notin E$. The smallest possible rank of such a matrix over $\F$ is denoted by $\minrk_\F(H)$. For the $5$-cycle we have $\minrk_\F(C_5) \leq \chibar(C_5) = 3$ by the linear clique-cover encoding and this is tight by as $\minrk_\F(C_5) \geq \lceil \beta(C_5) \rceil= 3$.
Finally, $\minrk_\F$ is clearly additive w.r.t.\ disjoint unions of graphs by its definition and thus $\minrk_\F(G) = 3k$ as required.
\end{proof}

\section{Approximating the Broadcast Rate}\svlabel{sec:approx}

This section is devoted to the proof of Theorem~\svref{thm-hypergraph-approx}, on polynomial-time algorithms for approximating $\beta$ and deciding whether $\beta=2$.  Working in the setting of a general broadcast network is somewhat delicate and we begin by sketching the arguments that will follow.

In the simpler case of undirected graphs, a $o(n)$-approximation to $\beta$ is implied by results of~\cites{Wigderson,BH,AKa} that together give a polynomial time procedure that finds either a small clique-cover or a large independent set (see Remark~\svref{rem-undirected-graphs}). To get an approximation for the general broadcasting problem we will apply a similar technique using analogues of independent sets and clique-covers that give lower and upper bounds respectively on the general broadcasting rate.  The analogue of an independent set is an \emph{expanding sequence} --- a sequence of receivers where the $i^{\mathrm{th}}$ receiver's desired message is unknown to receivers $1,\ldots,i-1$. The clique-cover analogue is a weak fractional hyperclique-cover (see Definition~\svref{def:hyperclique}).
In the remainder of this section, whenever we refer to hypercliques or hyperclique-covers we always mean weak hypercliques and weak hyperclique-covers.


We will prove that there is a polynomial time algorithm that outputs an expanding sequence of size $k$ or reports a fractional hyperclique-cover of size $O \left( k n^{1-1/k} \right)$; the approximation follows by setting $k$ appropriately.  We will argue that either we can partition the graph and apply induction or else the side-information map is dense enough to deduce existence of a small fractional hyperclique-cover.  The proof of the latter step deviates significantly from the techniques used for graphs, and seems interesting in its own right. We will give a simple procedure to randomly sample hypercliques and use it to produce a valid weight function for the hyperclique-cover by defining the weight of a hyperclique to be proportional to the probability it is sampled by the procedure.

To prove the second part of Theorem~\svref{thm-hypergraph-approx} we will prove that a structure called an \emph{almost alternating cycle} (AAC) constitutes a minimal obstruction to obtaining a broadcast rate of $2$.  The proof makes crucial use of Theorem~\svref{thm-hierarchy}, calculating the parameter $b_2$ for AAC's to  prove that their broadcast rate is strictly greater than $2$.  Furthermore, the proof reduces finding an AAC to finding the transitive closure of a particular relation, which is polynomial time computable.

\subsection{Approximating the broadcast rate in general networks}\svlabel{subsec:beta-approx}
We now present a  nontrivial approximation algorithm for $\beta$ for a general network described by a
hypergraph (that is, the most general framework where there are $m \geq n$ receivers).

\begin{remark}\svlabel{rem-undirected-graphs}
In the setting of undirected graphs a slightly better approximation algorithm for $\beta$ is a consequence of a result of Boppana and Halldorsson~\cite{BH}, following the work of Wigderson~\cite{Wigderson}. In~\cite{BH} the authors showed an algorithm that finds either a ``large'' clique or a ``large'' independent set in a graph (where the size guarantee involves the Ramsey number estimate). A simple adaptation of this result (Proposition~2.1 in the Alon-Kahale~\cite{AKa} work on approximating $\alpha$ via the $\vartheta$-function) gives a polynomial-time algorithm for finding an independent set of size $t_k(m)=\max\big\{ s : \binom{k+s-2}{k-1} \leq m\big\}$ in any graph satisfying $\chibar(G) \geq n/k + m$. In particular, taking $m=n/k$ with $k = \frac12 \log n$ we clearly have $t_k(m) \geq k$ for any sufficiently large $n$ and obtain that either $\chibar(G) < 4n/\log n$ or we can find an independent set of size $\frac12\log n$ in polynomial-time.
\end{remark}

We use the following notation: the $n$ message streams are identified with the elements of $[n] = V$.
The data consisting of the 
pairs $\{(N(j),f(j))\}_{j=1}^m$ is our \emph{
directed hypergraph} instance.
When referring to the hypergraph structure itself
(rather than the corresponding index coding problem) we will
refer to elements of $V$ as \emph{vertices} and we will refer to
pairs $(N(j),f(j))$ as \emph{directed hyperedges}.
For notational convenience,
we denote $S(j) = N(j) \cup \{f(j)\}$.



An \emph{expanding sequence} of size $k$ is a sequence of receivers
$j_1,\ldots,j_k$ such that
\begin{equation}
  \svlabel{eq-expanding-def}
  f(j_\ell) \not\in \bigcup_{i < \ell} S(i)
\end{equation}
for $1 \leq \ell \leq k.$
For a 
hypergraph $G$, let $\alpha(G)$
denote the maximum 
size of an expanding sequence.


\begin{lemma} \svlabel{lem:triv-lb}
Every 
hypergraph $G$ satisfies the bound
$\beta(G) \geq \alpha(G).$
\end{lemma}
\begin{proof}
The proof is by contradiction.  Let $j_1,\ldots,j_k$ be an expanding
sequence 
and suppose that there is an
index code that achieves rate $r < k.$  Let $J = \{j_1,\ldots,j_k\}.$
For $b = \log_2 |\Sigma|$ we have
\[
|\Sigma|^k = 2^{b k} > 2^{b r} \geq |\Sigma_P|.
\]
Let us fix an element $x^*_i \in \Sigma$ for every
$i \not\in \{f(j) : j \in J\},$ and define $\Psi$ to be the
set of all $\vec{x} \in \Sigma^n$ that satisfy
$x_i = x^*_i$ for all $i \not\in \{f(j) : j \in J\}.$
The cardinality of $\Psi$ is $|\Sigma|^k$, so the
Pigeonhole Principle implies that
the function $\Encode$, restricted to $\Psi$, is not
one-to-one.  Suppose that $\vec{x}$ and $\vec{y}$ are
two distinct elements of $\Psi$ such that $\Encode(\vec{x}) =
\Encode(\vec{y}).$  Let $i$ be the smallest index such
that $x_{f(j_i)} \neq y_{f(j_i)}.$  Denoting $j_i$ by $j$, we have
$x_{k} = y_{k}$ for all $k \in N(j)$,
because $N(j)$ does not contain
$f(j_{\ell})$ for any $\ell \geq i$, and the components with
indices $j_i, j_{i+1}, \ldots, j_k$ are the only components
in which $\vec{x}$ and $\vec{y}$ differ.  Consequently
receiver $j$ is unable to distinguish between message vectors
$\vec{x},\vec{y}$ even after observing the broadcast message,
which violates the condition that $j$ must be able to decode
message $f(j)$.
\end{proof}


\begin{lemma} \svlabel{lem:triv-ub}
Let $\weakhyp(G)$ denote the minimum weight of a
fractional weak hyperclique-cover of $G$.
Every 
hypergraph $G$ satisfies the bound
$\beta(G) \leq \weakhyp(G).$
\end{lemma}
\begin{proof}

The linear program defining $\weakhyp(G)$ has integer
coefficients, so $G$ has a fractional hyperclique
cover of weight $w = \weakhyp(G)$ in which the weight
$w(\hclq)$ of every hyperclique $\hclq$ is a
rational number.  Assume we are given such a fractional
hyperclique-cover, and choose an integer $d$ such that
$w(\hclq)$ is an integer multiple of $1/d$ for
every $\hclq$.  Let $\mathcal{C}$ denote a multiset
of hypercliques containing $d \cdot w(\hclq)$
copies of $\hclq$ for every hyperclique $\hclq$.
Note that the
the cardinality of $\mathcal{C}$ is $d \cdot w$.

For any hyperclique $\hclq$, let $f(\hclq)$
denote the set $\bigcup_{j \in \hclq} \{f(j)\}.$
For each $i \in [n]$, let $\mathcal{C}_i$ denote the
sub-multiset of $\mathcal{C}$ consisting of all
hypercliques $\hclq \in \mathcal{C}$ such that
$i \in f(\hclq)$.  Fix a finite field $\F$ such
that $|\F| > dw.$
Define $\Sigma = \F^d$ and
$\Sigma_P = \F^{d \cdot w}$.
Let $\{\xi_P^{\hclq}\}_{\hclq \in \mathcal{C}}$
be a basis for the dual vector space $\Sigma_P^*$ and
let $\{\xi_i^{\hclq}\}_{\hclq \in \mathcal{C}_i}$
be a set of dual vectors in $\Sigma^*$ such that
any $d$ of these vectors constitute a basis
for $\Sigma^*$.  (The existence of such a set of
dual vectors is guaranteed by our choice of $\F$ with
$|\F| > dw \geq d.$)

The encoding function is defined to be the unique
linear function satisfying
\[
\xi_P^{\hclq}(\Encode(x_1,\ldots,x_n)) =
\sum_{i \in f(\hclq)} \xi_i^{\hclq}(x_i)
\qquad \forall \hclq.
\]
For each receiver $j$, if $i=f(j)$, the set of dual vectors $\xi_i^{\hclq}$
with $j \in \hclq$ compose a basis of
$\Sigma^*$, hence to prove that $j$ can decode message $x_{i}$
it suffices to show that $j$ can determine the value of
$\xi_i^{\hclq}(x_{i})$ whenever $j \in \hclq.$
This holds because the public channel contains the value of
$\sum_{\ell \in f(\hclq)} \xi_{\ell}^{\hclq}(x_{\ell})$,
and receiver $j$ knows that value of $\xi_{\ell}^{\hclq}(x_{\ell})$
for every $\ell \neq i$ in $f(\hclq)$ because
$\ell \in N(j).$
\end{proof}

We now turn our attention to bounding the ratio
$\weakhyp(G)/\alpha(G)$ for a 
hypergraph $G$.  Our goal is to show that
this ratio is bounded by a function in $o(n).$
To begin with, we need an analogue of the lemma
that undirected graphs with small maximum degree
have small fractional chromatic number.

\begin{lemma} \svlabel{lem:low-degree}
If $G$ is 
a hypergraph with $n$ vertices,
and $d$ is a natural number such
that for every receiver $j$, $|S(j)| + d \geq n,$
then $\weakhyp(G) \leq 4d+2.$
\end{lemma}
\begin{proof}
Let us define a procedure for sampling a random subset
$T \subseteq [n]$ and a random hyperclique $\hclq$
as follows.  Let $\pi$ be a uniformly
random permutation of $[n+d]$, let $i$ be the least
index such that $\pi(i+1) > n,$  and let $T$ be the
set $\{\pi(1),\pi(2),\ldots,\pi(i)\}.$  (If $\pi(1)>n$
then $i=0$ and $T$ is the empty set.)  Now let
$\hclq$ be the set of all $j$ such that $f(j) \in T \subseteq S(j).$
(Note that $\hclq$ is indeed a hyperclique.)

For any hyperclique $\hclq$ let $p(\hclq)$ denote the
probability that $\hclq$ is sampled by this procedure
and let $w(\hclq) = (4d+2) \cdot p(\hclq).$  We claim
that the weights $w(\cdot)$ define a fractional
hyperclique-cover of $G$, or equivalently, that
for every receiver $j$, $\P(f(j) \in T \subseteq S(j)) \geq \frac{1}{4d+2}.$
Let $U(j)$ denote the set $\{f(j)\} \cup \left( [n] \setminus S(j) \right)
\cup \left( [n+d] \setminus [n] \right).$
The event $\mathcal{E} = \{f(j) \in T \subseteq S(j)\}$ occurs if and only
if, in the ordering of $U(j)$ induced by $\pi$, the first
element of $U(j)$ is $f(j)$ and the next element belongs
to $[n+d] \setminus [n].$  Thus,
\[
\P(\mathcal{E}) = \frac{1}{|U(j)|} \cdot
\frac{d}{|U(j)|-1}.
\]
The bound $\P(\mathcal{E}) \geq \frac{1}{4d+2}$ now follows
from the fact that $|U(j)| \leq 2d+1.$
\end{proof}

\begin{lemma} \svlabel{lem:hyper-wigderson}
If $G$ is 
a hypergraph and $\alpha(G) \leq k$,
then $\weakhyp(G) \leq 6k \cdot n^{1-1/k}.$  Moreover, there is a
polynomial-time algorithm, whose input is 
a hypergraph $G$
and a natural number $k$,
that either outputs an expanding sequence of size $k+1$ or
reports (correctly) that $\weakhyp(G) \leq 6k \cdot n^{1-1/k}.$
\end{lemma}
\begin{proof}
The proof is by induction on $k$.
In the base case $k=1$,
either $G$ itself is a hyperclique or
there is some pair of receivers $j,j'$ such
that $f(j)$ is not in $S(j')$.  In that case, the sequence
$j_1=j', j_2=j$ is an expanding sequence of size $2$.

For the induction step, for each hyperedge $j$ define the set
$D(j) = \{f(j)\} \cup \left( [n] \setminus S(j) \right)$
and let $j_1$ be a hyperedge such that $|D(j)|$ is maximum.
If $|D(j_1)| \leq n^{1 - 1/k}+1,$ then the bound
$|S(j)| + n^{1-1/k} \geq n$ is satisfied for
every $j$ and Lemma~\svref{lem:low-degree} implies
that $\weakhyp(G) < 4 n^{1-1/k} + 2 \leq 6 n^{1-1/k}.$
Otherwise, partition the vertex set of $G$
into $V_1 = [n] \setminus S(j_1)$ and $V_2 = S(j_1),$
and for $i=1,2$ define $G_i$ to be the
hypergraph with vertex set $V_i$ and edge
set $E_i$ consisting of all pairs $(N(j) \cap V_i, f(j))$
such that $(N(j),f(j))$ is a hyperedge of $G$
with $f(j) \in V_i.$  (We will call such a structure
the \emph{induced sub-hypergraph of $G$ on vertex
set $V_i$}.)
If $G_1$ contains an expanding sequence
$j_2,j_3,\ldots,j_{k+1}$ of size $k$, then
the sequence $j_1,j_2,\ldots,j_{k+1}$ is an
expanding sequence of size $k+1$ in $G$.
(Moreover, if an algorithm efficiently finds the
sequence $j_2,j_3,\ldots,j_{k+1}$ then it is easy
to efficiently construct the sequence $j_1,\ldots,j_{k+1}.$)
Otherwise, by the induction hypothesis,
$G_1$ has a fractional hyperclique-cover
of weight at most $6 (k-1) |V_1|^{1-1/(k-1)} \leq 6 (k-1) |V_1| n^{-1/k}$.
Continuing to process the induced sub-hypergraph
on vertex set $V_2$ in the same way, we arrive
at a partition of $[n]$ into disjoint vertex sets
$W_1, W_2, \ldots, W_{\ell}$ of cardinalities
$n_1,\ldots,n_{\ell}$, respectively, such that
for $1 \leq i < \ell$, the induced sub-hypergraph
on $W_i$ has a fractional clique-cover of weight
at most $6 (k-1) n_i n^{-1/k}$, and for $i=\ell$
the induced sub-hypergraph on $W_i$ satisfies
the hypothesis of Lemma~\svref{lem:low-degree}
with $d=n^{1-1/k}$
and consequently has a fractional hyperclique-cover
of weight at most $6 n^{1-1/k}.$
The lemma follows by summing the weights of these hyperclique-covers.
\end{proof}

\noindent Combining
Lemmas~\svref{lem:triv-lb},~\svref{lem:triv-ub},~\svref{lem:hyper-wigderson},
we obtain the approximation
algorithm asserted by Theorem~\svref{thm-hypergraph-approx}.


 \subsection{Extending the algorithm to networks with variable source rates}\svlabel{sec:weighted-hypergraph}
The aforementioned approximation algorithm for $\beta$ naturally extends to the setting where each source in the broadcast network has
its own individual rate. Namely, the $n$ message streams are identified with the elements of $[n] = V$, where message stream $i$ has a rate $r_i$, and the problem input consists of the vector $(r_1,\ldots,r_n)$ and the
pairs $\{(N(j),f(j))\}_{j=1}^m$. Thus the input is a \emph{weighted directed hypergraph} instance.
An index code for a weighted hypergraph consists of the following:
\begin{compactitem}
\item Alphabets $\Sigma_P$ and $\Sigma_i$ for $1 \leq i \leq n$,
\item An encoding function
$\Encode:  \prod_{i=1}^n \Sigma_i \rightarrow \Sigma_P$,
\item Decoding functions $\Decode_j : \Sigma_P \times \prod_{i \in N(j)} \Sigma_i
\rightarrow \Sigma_{f(j)}.$
\end{compactitem}
The encoding and decoding functions are required to satisfy
\[\Decode_j(\Encode(\sigma_1,\ldots,\sigma_n), \sigma_{N(j)}) = \sigma_{f(j)}\]
for all $j=1,\ldots,m$ and all $(\sigma_1,\ldots,\sigma_n) \in
\prod_{i=1}^n \Sigma_i.$  Here the notation
$\sigma_{N(j)}$ denotes the tuple obtained from a complete
$n$-tuple $(\sigma_1,\ldots,\sigma_n)$ by retaining only
the components indexed by elements of $N(j).$
An index code \emph{achieves} rate $r \geq 0$ if there exists
a constant $b>0$ such that $|\Sigma_i| \geq 2^{b\, r_i}$ for
$1 \leq i \leq n$ and $|\Sigma_P| \leq 2^{b\, r}.$  If so,
we say that rate $r$ is \emph{achievable}.  If $G$ is a
weighted hypergraph, we define
$\beta(G)$ to be the infimum of the set of achievable
rates.

The first step in generalizing the proof given in the previous subsection to the case where the $r_i$'s are non-uniform is to properly extend the notions of hypercliques and expanding sequences. A weak fractional hyperclique cover of a weighted hypergraph will now assign a weight $w(\hclq)$ to every weak hyperclique $\hclq$ such that for every receiver $j$, $\sum_{\hclq \ni j} w(\hclq) \geq r_{f(j)}$ (cf.\ Definition~\svref{def:hyperclique} corresponding to $r_{f(j)}=1$).
As before, the weight of a fractional weak hyperclique-cover is
given by $\sum_{\hclq} w(\hclq)$ and for a weighted hypergraph $G$ we let $\weakhyp(G)$
denote the minimum weight of a fractional weak hyperclique-cover.
An expanding sequence $j_1,\ldots,j_k$ is defined as before (see Eq.~\svref{eq-expanding-def}) except now we associate such a sequence with the weight $\sum_{\ell=1}^k r_{f(j_\ell)}$ and the quantity $\alpha(G)$ will denote
the maximum weight of an expanding sequence (rather than the maximum cardinality).

With these extended defintions, the proofs in the previous subsection carry unmodified to the weighted hypergraph setting with the single exception of Lemma~\svref{lem:hyper-wigderson}, where the assumption that the hypergraph is unweighted was essential to the proof. In what follows we will qualify an application of that lemma via a dyadic partition of the vertices of our weighted hypergraph according to their weights $r_i$.

 Assume without loss of generality that $0 \leq r_i \leq 1$ for
 every vertex $i \in [n]$, and partition the vertex of set
 $G$ into subsets $V_1, V_2, \ldots$ such that $V_s$
 contains all vertices $i$ such that $2^{-s} < r_i \leq 2^{1-s}.$
 Let $G_s$ denote the induced hypergraph on vertex
 set $V_s$.  For each of the nonempty hypergraphs $G_s$,
 run the algorithm in Lemma~\svref{lem:hyper-wigderson} for
 $k=1,2,\ldots$ until the smallest value of $k(s)$ for which
 an expanding sequence of size $k(s)+1$ is not found.
 If $G_s^{\circ}$ denotes the unweighted version of
 $G_s$, then we know that
 \begin{align*}
 \alpha(G_s) & \geq 2^{-s} \alpha(G_s^{\circ}) \geq 2^{-s} k(s) \\
 \weakhyp(G_s) & \leq 2^{1-s} \weakhyp(G_s^{\circ}) \leq 2^{-s} \cdot
 12 k(s) n^{1-1/k(s)}.
 \end{align*}
 In addition, for each $i \in V_s$ the set of hyperedges
 containing $i$ constitutes a hyperclique, which implies
 the trivial bound
 \[
 \weakhyp(G_s) \leq \sum_{i \in V_s} r_i \leq 2^{1-s} |V_s|.
 \]
 Combining these two upper bounds for $\weakhyp(G_s)$, we obtain
 an upper bound for $\weakhyp(G)$:
 \begin{equation} \svlabel{eq:tau}
 \weakhyp(G) \leq \sum_{s=1}^{\infty} \weakhyp(G_s) \leq
 \sum_{s=1}^{\infty} 2^{-s} \cdot \min \left\{
 12 k(s) n^{1-1/k(s)}, \, 2 |V_s| \right\}.
 \end{equation}
 We define $\tau(G)$ to be the right side of \sveqref{eq:tau}.
 We have described a polynomial-time algorithm to compute $\tau(G)$
 and have justified the relation $\weakhyp(G) \leq \tau(G)$,
 so it remains to show that $\tau(G) / \alpha(G) \leq
 c n \left( \frac{\log \log n}{\log n} \right)$ for some constant $c$.

 The bound $\tau(G) \leq n$ follows immediately
 from the definition of $\tau$, so if $\alpha(G) \geq
 \frac{\log n}{\log \log n}$ there is nothing to prove.
 Assume henceforth that $\alpha(G) < \frac{\log n}{\log \log n}$,
 and define $w$ to be the smallest integer such that
 $2^w \cdot \alpha(G) > \frac{\log n}{2 \log \log n}.$
 We have
 \begin{align}
 \nonumber
 \tau(G) &\leq
 \sum_{s=1}^{w} 2^{-s} \cdot 12 k(s) n^{1-1/k(s)} \;+\;
 \sum_{s=w+1}^{\infty} 2^{1-s} \cdot |V_s| \\
 \nonumber
 & \leq
 12 n \sum_{s=1}^{w} 2^{-s} k(s) n^{-1/k(s)} \;+\;
 2^{-w} \cdot n \\
 \svlabel{eq:tau.2}
 & <
 12 n \alpha(G) \sum_{s=1}^{w} n^{-1/k(s)} \;+\;
 2 n \alpha(G) \left( \frac{\log \log n}{\log n} \right),
 \end{align}
 with the last line derived using the relations
 $2^{-s} k(s) \leq \alpha(G_s) \leq \alpha(G)$
 and $2^{-w} < \alpha(G) \big( \frac{2 \log \log n}{\log n} \big)$.
 Applying once more the fact that $2^{-s} k(s) \leq \alpha(G)$,
 we find that $n^{-1/k(s)} \leq n^{-1/\left(2^s \cdot \alpha(G)\right)}.$
 Substituting this bound into \sveqref{eq:tau.2} and letting $\alpha$
 denote $\alpha(G)$, we have
 \[
 \frac{\tau(G)}{\alpha(G)} \leq
 2 n \left( \frac{\log \log n}{\log n} \right) \;+\;
 12 n \left( n^{-1/2\alpha} + n^{-1/4\alpha} + \cdots + n^{-1/2^w \alpha}
 \right).
 \]
 In the sum appearing on the right side, each term is the square of
 the one following it.  It now easily follows
 that the final term in
 the sum is less than $1/2$, so the entire sum is bounded above
 by twice its final term.  Thus
 \begin{equation} \svlabel{eq:tau.3}
 \frac{\tau(G)}{\alpha(G)} \leq
 2 n \left( \frac{\log \log n}{\log n} \right) \;+\;
 24 n \cdot n^{-1/2^w \alpha}.
 \end{equation}
 Our choice of $w$ ensures that $2^w \alpha \leq \frac{\log n}{\log \log n}$
 hence
 $n^{-2^{-w} a} \leq n^{-\log \log n / \log n} = (\log n)^{-1}$.
 By substituting this bound into \sveqref{eq:tau.3} we obtain
 \[
 \frac{\tau(G)}{\alpha(G)} \leq n
 \left( \frac{2 \log \log n}{\log n} + \frac{24}{\log n} \right)\,,
 \]
 as desired.

\subsection{Proof of Theorem~\svref{thm-hypergraph-approx}, determining whether the broadcast rate equals 2}\svlabel{subsec:beta-equals-2}

Let $G$ be an undirected graph with independence number $\alpha=2$. Clearly, if $\overline{G}$ is bipartite than $\chibar(G) = 2$ and so $\beta(G) = 2$ as well. Conversely, if $\overline{G}$ is not bipartite then it contains an odd cycle, the smallest of which is induced and has $k\geq 5$ vertices since the maximum clique in $\overline{G}$ is $\alpha(G)=2$. In particular,
\svapdx{Theorem~\svref{thm-beta-exact}}{Theorem~\svref{thm-cycles}} implies that $\beta(G) \geq \beta(\overline{C_k}) = \frac{k}{\lfloor k/2\rfloor} > 2$. We thus conclude the following:

\begin{corollary}\svlabel{cor-beta-2-undir}
Let $G$ be an undirected graph on $n$ vertices whose complement $\overline{G}$ is nonempty. Then $\beta(G) = 2$ if and only if $\overline{G}$ is bipartite.
\end{corollary}

A polynomial time algorithm for determining whether $\beta=2$ in undirected graphs follows as an immediate consequence of Corollary~\svref{cor-beta-2-undir}. However, for broadcasting with side information in general --- or even for the special case of directed graphs (the main setting of~\cites{BK,BBJK}) --- it is unclear whether such an algorithm exists.  In this section we provide such an algorithm, accompanied by a characterization theorem that generalizes the above characterization for undirected graphs.  To state our characterization we need the following definitions.  As in Section~\svref{subsec:beta-approx} we use $S(j)$ to denote the set $N(j) \cup \{f(j)\}$.  Additionally, we introduce the notation $T(j)$ to denote the complement of $S(j)$ in the set of messages.  When referring to the message desired by receiver $R_j$, we abbreviate $x_{f(j)}$ to $x(j)$.
\svapdx{}{Henceforth, when referring to a hypergraph $G=(V,E)$, we assume that for each edge $j \in E$, the hypergraph structure specifies the vertex $f(j)$ and both of the sets $S(j),T(j)$.}
\begin{definition} \svlabel{def:h-compat}
If $G=(V,E)$ is a directed hypergraph and $S$ is a set,
a function $F : V \to S$ is said to be \emph{$G$-compatible}
if for every edge $j \in E$, there are two \emph{distinct}
elements $t,u \in S$ such that $F$ maps every element of $T(j)$ to $t$,
and it maps $f(j)$ to $u$.
\end{definition}

\begin{definition} \svlabel{def:aac}
If $G=(V,E)$ is a directed hypergraph, an \emph{almost alternating
(2n+1)-cycle} in $G$ is
a sequence of vertices
$v_{-n}, v_{-n+1}, \ldots, v_n$,
and a sequence of edges
$j_0, \ldots, j_n$,
such that for $i=0,\ldots,n$, the vertex $f(j_i)$ is equal to $v_{i-n}$
and the set $T(j_i)$ contains $v_{i}$, as well as $v_{i+1}$ if $i < n$.
\end{definition}

\begin{theorem}\svlabel{thm-directed-beta-2}
For a directed hypergraph $G$ the following are equivalent:
\begin{compactenum}[(i)]
\item  $\beta(G)=2$ \svlabel{aac:beta}
\item  There exists a set $S$ and a $G$-compatible function $F : V \to S$.
\svlabel{aac:g-compat}
\item  $G$ contains no almost alternating cycles.
\svlabel{aac:aac}
\end{compactenum}
Furthermore there is a polynomial-time algorithm to decide
if these equivalent conditions hold.
\end{theorem}
\begin{proof}
%
{\bf(\svref{aac:beta})}$\Rightarrow${\bf(\svref{aac:aac}):}
The contrapositive statement says that if $G$ contains an almost alternating
cycle then $\beta(G) > 2.$  Let $v_{-n},\ldots,v_n$ be the vertices
of an almost alternating $(2n+1)$-cycle with edges $j_0,\ldots,j_n$.
\svapdx{
We manipulate entropy inequalities to
prove that $b_2(G) > 2$ whenever $G$ contains an
almost alternating cycle; the details are in
Appendix~\ref{subsec:beta-equals-2}.  Theorem~\svref{thm-hierarchy}
then implies that $\beta(G) > 2$.
}{
To prove $\beta(G) > 2$ we manipulate
entropy inequalities involving the random variables
$\{x_i : -n \leq i \leq n\}$ and $y$, where $x_i$ denotes
the message associated to vertex $v_i$ normalized to have entropy $1$, and $y$ denotes
the public channel.  For $S \subseteq \{y,x_{-n},\ldots,x_n\}$,
let $H(S)$ denote the entropy of the joint distribution of
the random variables in $S$, and let $\hb{S}$ denote
$H(\overline{S})$.  Let $S_{i:j}$ denote the set $\{x_i,x_{i+1},\ldots,x_j\}.$

For $0 \leq i \leq n-1$, we have
\begin{equation} \svlabel{eq:aac2}
H(y) + (2n-2) \geq
\hb{\{x_{i-n},x_{i},x_{i+1}\}} =
\hb{\{x_{i},x_{i+1}\}} =
\hb{S_{i:i+1}}\,,
\end{equation}
where the second equation holds because receiver
$j_i$ can decode message $x_{i-n} = x(j_i)$
given the value $y$ and the values $x_k$ for $k \in N(j_i)$.
Using submodularity we have that for $0 < j < n,$
\begin{equation} \svlabel{eq:aac3}
\hb{S_{0:j}} + \hb{S_{j:j+1}} \geq
\hb{S_{0:j+1}} + \hb{x_j} =
\hb{S_{0:j+1}} + \hb{\emptyset} =
\hb{S_{0:j+1}} + 2n+1 \,.
\end{equation}
Summing up \sveqref{eq:aac3} for $j=1,\ldots,n-1$ and
canceling terms that appear on both sides, we obtain
\begin{equation} \svlabel{eq:aac4}
\sum_{j=0}^{n-1} \hb{S_{j:j+1}} \geq \hb{S_{0:n}} + (n-1)(2n+1)\,.
\end{equation}
Summing up \sveqref{eq:aac2} for $i=0,\ldots,n-1$ and
combining with \sveqref{eq:aac4} we obtain
\begin{equation} \svlabel{eq:aac5}
n H(y) + n(2n-2) \geq \hb{S_{0:n}} + (n-1)(2n+1)\,.
\end{equation}
Now, observe that
\begin{equation} \svlabel{eq:aac6}
\hb{S_{0:n}} + n-1 \geq
\hb{x_0,x_n} \geq \hb{x_n} \geq \hb{\emptyset} = 2n+1 \,.
\end{equation}
Summing \sveqref{eq:aac5} and \sveqref{eq:aac6}, we obtain
\begin{align*}
n H(y) + 2n^2 - n - 1  &\geq 2n^2 + n 
\end{align*}
and rearranging we get $H(y)\geq 2+n^{-1}$,
from which it follows that $\beta(G) \geq 2 + n^{-1}$.
}

{\bf (\svref{aac:aac})$\Rightarrow$(\svref{aac:g-compat}):}
Define a binary relation $\sharp$ on the vertex set $V$
by specifying that $v \sharp w$ if there exists an
edge $j$ such that $\{v,w\} \subseteq T(j)$. 
Let $\sim$ denote the transitive closure of $\sharp$.
Define
$F$ to be the quotient map from $V$ to
the set $S$ of equivalence classes of $\sim$.
We need to check that $F$ is $G$-compatible.
For every edge $j \in E$, the definition of relation $\sharp$
trivially implies that $F$ maps all of $T(j)$ to a single element of $S$.
The fact that it maps $f(j)$ to a \emph{different} element of $S$
is a consequence of
the non-existence of almost  alternating cycles.
A relation $f(j) \sim v$ for some $v \in T(j)$ would imply
the existence of a sequence $v_0, \ldots, v_n$ such that
$v_0=f(j), v_n=v,$ and $v_i \sharp v_{i+1}$ for
$i=0,...,n-1$.  If we choose $j_i$ for $0 \leq i < n$ to
be an edge such that $T(j_i)$ contains $v_i,v_{i+1}$
(such an edge exists because $v_i \sharp v_{i+1}$)
and we set $j_n = j$ and $v_{i-n} = f(j_i)$ for $i=0,\ldots,n-1$,
then the vertex sequence $v_{-n},\ldots,v_n$ and edge
sequence $j_0,\ldots,j_n$ constitute an almost alternating
cycle in $G$.

Computing the relation $\sim$ and the function $F$,
as well as testing that $F$ is $G$-compatible, can easily
be done in polynomial time, implying the final sentence
of the theorem statement.

{\bf (\svref{aac:g-compat})$\Rightarrow$(\svref{aac:beta}):}
If $F : V \to S$ is $G$-compatible,
we may compose $F$ with a one-to-one mapping from $S$
into a finite field $\F$, to obtain a function
$\phi :  V \to \F$ that is $G$-compatible.
The public channel broadcasts two elements of $\F$, namely:
\svapdx{
\[
y = \sum_v x_v, \qquad z = \sum_v \phi(v) x_v.
\]
}{
\begin{align*}
  y & = \sum_v  x_v \\
  z & = \sum_v  \phi(v) x_v
\end{align*}
}
Receiver $R_j$ now decodes message $x(j)$ as follows.
Let $c$ denote the unique element of $\F$ such that $\phi(v)=c$ for every
$v$ in $T(j)$.
Using the pair $(y,z)$ from the public channel, $R_j$
can form the linear combination
\svapdx{$}{$$}
cy - z = \sum_v [c - \phi(v)] x_v.
\svapdx{$}{$$}
We know that every $v \in T(j)$ appears with coefficient zero in this sum.
For every $v \in N(j)$, receiver $R_j$ knows the value of $x_v$ and can
consequently subtract off the term $[c-\phi(v)] x_v$ from the sum.
The only remaining term is $[c-\phi(x(j))] x(j)$.
The coefficient $c-\phi(x(j))$ is nonzero, because $\phi$ is $G$-compatible.
Therefore $R_j$ can decode $x(j)$.
\end{proof}


\section{The gap between the broadcast rate and clique cover numbers}
\subsection{Separating the broadcast rate from the extreme LP solution $b_n$}\svlabel{sec:sep-beta-alpha}
In this section we prove Theorem~\svref{thm-beta-chif-gap} that shows a strong form of separation between $\beta$ and its upper bound $b_n = \chibar_f$. Not only can we have a family of graphs where $\beta=O(1)$ while $\chibar_f$ is unbounded, but one can construct such a family where $\chibar_f$ grows polynomially fast with $n$.

\begin{proof}[\textbf{\emph{Proof of Theorem~\svref{thm-beta-chif-gap}}}]
The following family of graphs (up to a small modification) was introduced by Erd\H{o}s and R\'enyi in~\cite{ER}. Due to its close connection to the (Sylvester-)Hadamard matrices when the chosen field has characteristic 2 we refer to it as the \emph{projective-Hadamard} graph $H(\F_q)$:
\begin{compactenum}
\item  Vertices are the non-self-orthogonal vectors in the $2$-dimensional projective space over $\F_q$.
\item Two vertices are adjacent iff their corresponding vectors are non-orthogonal.
\end{compactenum}
Let $q$ be a prime-power. We claim that the projective-Hadamard graph $H(\F_q)$ on $n=n(q)$ vertices satisfies $\beta = 3$ while $\chibar_f = \Theta(n^{1/4})$. The latter
is a well-known fact which appears for instance in~\cites{AK,MW}. Showing that $\chibar_f \geq (1-o(1))n^{1/4}$ is straightforward and we include an argument establishing this for completeness.

The fact that $\beta \geq 3$ follows from the fact that the standard basis vectors form an independent set of size $3$. A matching upper bound will follow from the $\minrk_\F$ parameter defined in Section~\svref{subsec:cor-of-thm-1}:
Let $\F$ be some finite field and let $\ell=\minrk_\F(G)$ be the length of the optimal linear encoding over $\F$ for the Index Coding problem of a graph $G$ with messages taking values in $\F$. Broadcasting $\ell \lceil \log_2 |\F| \rceil$ bits allows each receiver to recover his required message in $\F$ and so clearly $\beta \leq \ell$. It thus follows that $\lceil \beta(G) \rceil \leq \minrk_\F(G)$ for any graph $G$ and finite field $\F$.

Here, dealing with the projective-Hadamard graph $H$, let $B$ be the Gram matrix over $\F_q$ of the vectors corresponding to the vertices of $H$. By definition the diagonal entries are nonzero and whenever two vertices $u,v$ are nonadjacent we have $B_{uv} = 0$. In particular $B$ is a representation for $H$ over $\F_q$ which clearly has rank $3$ as the standard basis vectors span its entire row space. Altogether we deduce that $\beta(H) = 3$ whereas $\chibar_f = \Theta(n^{1/4})$, as required.

The fact that $\chibar_f \geq (1-o(1))n^{-1/4}$ will follow from a straightforward calculation showing that the clique-number of $H$ is at most $(1+o(1))q^{3/2} = (1+o(1))n^{3/4}$.

Consider the following multi-graph $G$ which consists of the entire projective space:
\begin{compactenum}
  \item Vertices are all vectors of the $2$-dimensional projective space over $\F_q$.
  \item Two (possibly equal) vertices are adjacent iff their corresponding vectors are orthogonal.
\end{compactenum}
Clearly, $G$ contains the complement of the Hadamard graph $H(\F_q)$ as an induced subgraph
and it suffices to show that $\alpha(G) \leq (1+o(1))q^{3/2}$.

It is well-known (and easy) that $G$ has $N=q^2+q+1$ vertices and that every vertex of $G$ is adjacent to precisely $q+1$ others. Further observe that for any $u,v \in V(G)$ precisely one vertex of $G$ belongs to $\{u,v\}^\bot$ (as $u,v$ are linearly independent vectors). In other words, the codegree of any two vertices in $G$ is $1$.
We conclude that $G$ is a strongly-regular graph (see e.g.~\cite{GR} for more details on this special class of graphs) with codegree parameters $\mu=\nu=1$
(where $\mu$ is the codegree of adjacent pairs and $\nu$ is the codegree of non-adjacent ones).
There are thus precisely 2 nontrivial eigenvalues of $G$ given by
$\frac12 ((\mu-\nu)\pm\sqrt{(\mu-\nu)^2+4(q+1-\nu)})= \pm\sqrt{q}$,
and in particular the smallest eigenvalue is $\lambda_N=-\sqrt{q}$. Hoffman's eigenvalue bound (stating that $\alpha \leq \frac{-m\lambda_m}{\lambda_1-\lambda_m}$ for any regular $m$-vertex graph with largest and smallest eigenvalues $\lambda_1,\lambda_m$ resp., see e.g.~\cite{GR}) now shows
\begin{equation*}
\alpha(G) \leq \frac{-N \lambda_N}{(q+1)-\lambda_N} = \frac{(q^2+q+1)\sqrt{q}}{q-\sqrt{q}+1}
= q^{3/2}+q+\sqrt{q}\,,
\end{equation*}
as required.
\end{proof}

In addition to demonstrating a large gap between $\chibar_f$ and $\beta$ on the projective-Hadamard graphs, we show that even in the extreme cases where $G$ is a triangle-free graph on $n$ vertices, in which case $\chibar_f(G) \geq n/2$, one can
construct Index Coding schemes that significantly
outperform $\chibar_f$. We prove this
in Section~\svref{subsec:sep-triangle-free} by providing a
family of triangle-free graphs on $n$ vertices where
$\beta \leq \frac38 n$.

\subsection{Broadcast rates for triangle-free graphs}\svlabel{subsec:sep-triangle-free}

In this section we study the behavior of the broadcast rate for triangle-free graphs, where the upper bound $b_n$ on $\beta$ is at least $n/2$. The first question in this respect is whether possibly $\beta = b_n$ in this regime, i.e.\
for such sparse graphs one cannot improve upon the fractional clique-cover approach for broadcasting. This is answered by the following result.
\begin{theorem}\svlabel{thm-triangle-free}
There exists an explicit family of triangle-free graphs on $n$ vertices where $\chibar_f \geq n/2$
whereas the broadcast rate satisfies $\beta \leq \frac38 n$.
\end{theorem}
The following lemma will be the main ingredient in the construction:
\begin{lemma}\svlabel{lem-triangle-free-family}
For arbitrarily large integers $n$ there exists a family $\cF$ of subsets of $[n]$ whose size is at least $8n/3$ and has the following two properties:
\begin{inparaenum}[(i)]
  \item Every $A \in \cF$ has an odd cardinality.
  \item There are no distinct $A,B,C\in\cF$ that have pairwise odd cardinalities of intersections.
\end{inparaenum}
\end{lemma}
\begin{remark}
For $n$ even, a simple family $\cF$ of size $2n$ with the above properties
is obtained by taking all the singletons and all their complements. However, for our application here it is crucial to obtain a family $\cF$ of size strictly larger than $2n$.
\end{remark}
\begin{remark}
  The above lemma may be viewed as a higher-dimensional analogue of the Odd-Town theorem: If we consider a graph on the odd subsets with edges between those with an odd cardinality of intersection, the original theorem looks for a maximum independent set while the lemma above looks for a maximum triangle-free graph.
\end{remark}
\begin{proof}[Proof of lemma]
It suffices to prove the lemma for $n=6$ by super-additivity (we can partition a ground-set $[N]$ with $N=6m$ into disjoint $6$-tuples and from each take the original family $\cF$).

Let $U_1 = \big\{ \{x\} : x \in [5]\big\}$ be all singletons except the last, and $U_2 = \big\{ A \cup \{6\} : A \subset [5]\,,\,|A|=2 \big\}$.
Clearly all subsets given here are odd.

We first claim that there are no triangles on the graph induced on $U_2$. Indeed, since all subsets there contain the element $6$, two vertices in $U_2$ are adjacent iff their corresponding 2-element subsets $A,A'$ are disjoint, and there cannot be 3 disjoint 2-element subsets of $[5]$.

The vertices of $U_1$ form an independent set in the graph, hence the only remaining option for a triangle in the induced subgraph on $U_1\cup U_2$ is of the form $\{x\}, (A\cup\{6\}), (A'\cup \{6\})$. However, to support edges from $\{x\}$ to the two sets in $U_2$ we must have that $x$ belongs to both sets, and since $x\neq 6$ by definition we must have $x \in A \cap A'$. However, we must also have $A \cap A' = \emptyset$ for the two vertices in $U_2$ to be adjacent, contradiction.

To conclude the proof observe that adding the extra set $[5]$ does not introduce any triangles, since $U_1$ is an independent set while $[5]$ is not adjacent to any vertex in $U_2$ (its intersection with any set $(A\cup\{6\}) \in U_2$ contains precisely 2 elements). Altogether we have $|\cF| = 5+\binom{5}2 + 1=\frac83 n$.
\end{proof}
\begin{proof}[\emph{\textbf{Proof of Theorem~\svref{thm-triangle-free}}}]
Let $\cF$ be the family provided by the above lemma and consider the graph $G$ whose $N$ vertices are the elements of $\cF$ with edges between $A,B$ whose cardinality of intersection is odd. By definition the graph $G$ is triangle-free and we have $\chibar_f(G) \geq N/2$.

Next, consider the binary matrix $M$ indexed by the vertices of $G$ where $M_{A,B} = |A\cap B|\pmod{2}$. All the diagonal entries of $M$ equal $1$ by the fact that $\cF$ is comprised of odd subsets only, and clearly $M$ is a representation of $G$ over $GF(2)$. At the same time, $M$ can be written as $F F^\mathrm{T}$ where $F$ is the $N\times n$ incidence-matrix of the ground-set $[n]$ and  subsets of $\cF$. In particular we have that $\rank(M) \leq \rank(F) \leq n$ over $GF(2)$.
This implies that $\minrk_2(G) \leq n$ and the proof is now concluded by the fact that $\beta(G) \leq \minrk_2(G)$.
\end{proof}

\begin{remark} The construction of the family of subsets $\cF$ in Lemma~\svref{lem-triangle-free-family} relied on a triangle-free 15-vertex base graph $H$ which is equivalent to the Peterson graph with 5 extra vertices added to it, each one adjacent to one of the independent sets of size 4 in the Peterson graph.
\end{remark}

Having discussed the relation between $\beta$ and $b_n$ for sparse graphs we now turn our attention to the analogous question for the other extreme end, namely whether $\beta=b_1$ when $b_1=\alpha$ attains its smallest possible value (other than in the complete graph) of 2.

\subsection{Graphs with a broadcast rate of nearly 2}

We now return to the setting of undirected graphs, where the class of $\{ G : \beta(G) = 2\}$ is simply the complements of nonempty bipartite graphs, where in particular Index Coding is trivial.
It turns out that extending this class to $\{G : \beta(G) < 2+\epsilon\}$ for any fixed small $\epsilon> 0$ already turns this family of graphs to a much richer one, as the following simple corollary of Theorem~\svref{thm-hierarchy} shows. Recall that the Kneser graph with parameters $(n,k)$ is the graph whose vertices are all the $k$-element subsets of $[n]$ where two vertices are adjacent iff their two corresponding subsets are disjoint.
\begin{corollary}\svlabel{cor-kneser}
Fix $0 < \epsilon < \frac12$ and let $G$ be the complement of the Kneser$(n,k)$ graph on $N=\binom{n}{k}$ vertices for $n = (2+\epsilon)k$. Then $\beta(G) \leq 2+\epsilon$ whereas $\chibar(G) \geq (\epsilon/2)\log N$.
\end{corollary}
\begin{proof}
  Using topological methods, Lov\'asz~\cite{Lovasz} proved that
   the Kneser graph with parameters $(n,k)$ has chromatic number $n-2k+2$, in our case giving
   that $\chibar(G) = \epsilon k + 2 \leq (\epsilon/2)\log N$ (with the last inequality due to the fact that $N \geq [\mathrm{e}(2+\epsilon)]^k$ and so $k \geq \frac12 \log N$).
  At the same time, it is well known that $G$ satisfies $\chibar_f = n/k$ (its maximum clique corresponds to a maximum set of intersecting $k$-subsets, which has size $\omega=\binom{n-1}{k-1}$ by the Erd\H{o}s-Ko-Rado Theorem, and being vertex-transitive it satisfies $\chibar_f = N / \omega$).
  The bound $\beta \leq b_n = \chibar_f$ given in Theorem~\svref{thm-hierarchy} thus completes the proof.
\end{proof}

\section{Establishing the exact broadcast rate for families of graphs}\svlabel{sec:beta-of-graphs}

\subsection{The broadcast rate of cycles and their complements}
The following theorem establishes the value of $\beta$ for cycles and their complements via the LP framework of Theorem~\svref{thm-hierarchy}.
\begin{theorem}\svlabel{thm-cycles}
  For any integer $n\geq 4$ the $n$-cycle satisfies $\beta(C_{n}) = n/2$ whereas its complement satisfies $\beta(\overline{C_n})=n/\lfloor n/2\rfloor $. In both cases $\beta_1 = \lceil\beta \rceil$ while
  $\alpha = \lfloor \beta \rfloor$.
\end{theorem}
\begin{proof}
As the case of $n$ even is trivial with all the inequalities in~\sveqref{eq-trivial-ineqs} collapsing into an equality (which is the case for any perfect graph), assume henceforth that $n$ is odd.
We first show that $\beta(C_n) = n/2$. Putting $n=2k+1$ for $k \geq 2$, we aim to prove that $b_2 \geq k+1/2 $, which according to Theorem~\svref{thm-hierarchy} will imply the required result since clearly $\chibar_f = k+1/2$.

Denote the vertices $V$ of the cycle by $0,1,\ldots,2k$.  Further define:
\begin{align*}
  E &= \{i \;:\; i\equiv0\bmod{2}~,~ i \ne 2k\}  &\mbox{(Evens)}\,, \\
  O &= \{i\;:\; i \equiv1\bmod{2}\} &\mbox{(Odds)}\,, \\
  E^+ &= \{i \;:\; i \le 2k-2\} &\mbox{(Evens decoded)}\,,\\
   O^+ &= \{i \;:\; 1 \le i \le 2k-1\} &\mbox{(Odds decoded)}\,,\\
    M &= \{i \;:\; 1 \le i \le 2k-2\} &\mbox{(Middle)}\,.
\end{align*}
Next, consider the following constraints in the LP $\mathcal{B}_2$:
\begin{align*}
X(\emptyset) +  k &\ge  X(E) & \mbox{(slope)}\\
X(\emptyset) +  k &\ge  X(O) & \mbox{(slope)}\\
X(\emptyset) +  1 &\ge  X(\{2k\}) & \mbox{(slope)}\\
X(E) & \ge X(E^+) & \mbox{(decode)}\\
X(O) & \ge X(O^+) & \mbox{(decode)}\\
X(E^+) + X(O^+) & \ge X(V) + X(M) & \mbox{(submod , decode)}\\
X(M) + X(\{2k\}) &\ge X(V) + X(\emptyset) & \mbox{(submod , decode)}\\
2X(V) &\ge  2(2k+1)  & \mbox{(initialize)}&\,.
\end{align*}

Summing and canceling we get $2X(\emptyset) + 2k+1 \ge 4k+2$, implying $X(\emptyset) \ge k+1/2$.  The main idea of this proof, as with the ones to follow, is that we input some sets of vertices and then apply decoding to the sets as well as combine them together using submodularity to eventually output $X(V)$ and $X(\emptyset)$.

It remains to treat complements of odd cycles. Let $H=\overline{\aac_{n}}$ be the complement of a directed odd almost-alternating cycle on $n$ vertices (as defined in Section~\svref{subsec:beta-equals-2}). Treating $\overline{C_{n}}$ as a directed graph (replacing each edge with a bi-directed pair of edges) it is clearly a spanning subgraph of $H$, hence $\beta(\overline{C}_{n})$ is at least as large as $\beta(H)$. The proof in Section~\svref{subsec:beta-equals-2} establishes that $\beta(H) \geq \frac{n}{\lfloor n/2\rfloor}$, translating to a lower bound on $\beta(\overline{C_n})$. The matching upper bound follows from the fact that due to Theorem~\svref{thm-hierarchy} we have $\beta(\overline{C_n}) \leq \chibar(\overline{C_n}) = \frac{n}{\lfloor n/2\rfloor}$.
\end{proof}

\subsection{The broadcast rate of cyclic Cayley Graphs}\svlabel{sec:ap-cayley}

In this section we demonstrate how the same framework of the proof of Theorem~\svref{thm-cycles} may be applied with a considerably more involved sequence of entropy-inequalities to establish the broadcast rate of two classes of Cayley graphs of the cyclic group $\Z_n$. Recall that a \emph{cyclic Cayley graph} on $n$ vertices with a set of generators $G \subseteq \{1,2,\ldots,\lfloor n/2\rfloor\}$ is the graph on the vertex set $\{0,1,2,\ldots,n-1\}$ where $(i,j)$ is an edge iff $j - i \equiv g \pmod{n}$ for some $g \in G$.

\begin{theorem}\svlabel{thm-3-reg-cayley}
For any $n\geq 4$, the 3-regular Cayley graph of $\Z_n$ has broadcast rate $\beta = n/2$.
\end{theorem}
\begin{theorem}[Circulant graphs]\svlabel{thm-circulant}
For any integers $n\geq 4$ and $k < \frac{n-1}2$, the Cayley graph of $\Z_n$ with generators $\{\pm1,\ldots,\pm k\}$ has broadcast rate $\beta=n/(k+1)$.
\end{theorem}

To simplify the exposition of the proofs of these theorems we make use of the following definition.
\begin{definition}
A \emph{slice} of size $i$ in $\Z_n$ indexed by $x$ is the subset of $i$ contiguous vertices on the cycle given by $\{x+j \pmod{n} : 0 \le j < i\}$.
\end{definition}

\begin{proof}[\textbf{\emph{Proof of Theorem~\svref{thm-3-reg-cayley}}}]
It is not hard to see that for a cyclic Cayley graph to be $3$-regular it must have two generators, $1$ and $n/2$, and $n$ must be even.  If $n$ is not divisible by four, then it is easy to check that there is an independent set of size $n/2$ and $\chibar_f$ is also $n/2$.  Thus, it immediately follows that $\beta = n/2$.
For 3-regular cyclic Cayley graphs where $n$ is divisible by four, $\alpha$ is strictly less than $n/2$.  So to prove that $\beta = n/2$ we use the LP $\mathcal{B}_2$ to show $b_2 \ge n/2$, implying $\beta \ge n/2$.

Let $0,1,2,\ldots,4k-1$ be the vertex set of the graph. We assume that any solution $X$ has cyclic symmetry.  That is, $X(S) = X(\{s+i|s \in S\})$ for all $i \in [0,4k-1]$.  This assumption is without loss of generality because we can take any LP solution $X$ and find a new one $X'$ that is symmetric and has the same value by setting $X'(S) = \frac{1}{4k}\sum_{i = 0}^{4k-1} X(\{s+i|s \in S\})$.  All the constraints are feasible for $X'$ because each is simply the average of $4k$ feasible constraints.

In our proof we will be using the following subsets of vertices:
\begin{align*}
[i] &= \{0,1,2,\ldots,i-1\} \;\; \text{(a slice of size $i$)}\\
D &= \{0,2, \ldots, 2k-4, 2k-2, 2k+1,2k+3, \ldots , 4k-5,4k-3\}\\
D^+ &= \{0,1,2,\ldots,2k-4, 2k-3, 2k-2, 2k+1, 2k+2,2k+3, \ldots, 4k-4, 4k-3\}\,.
\end{align*}
\begin{figure}[tb]
\begin{center}
\includegraphics{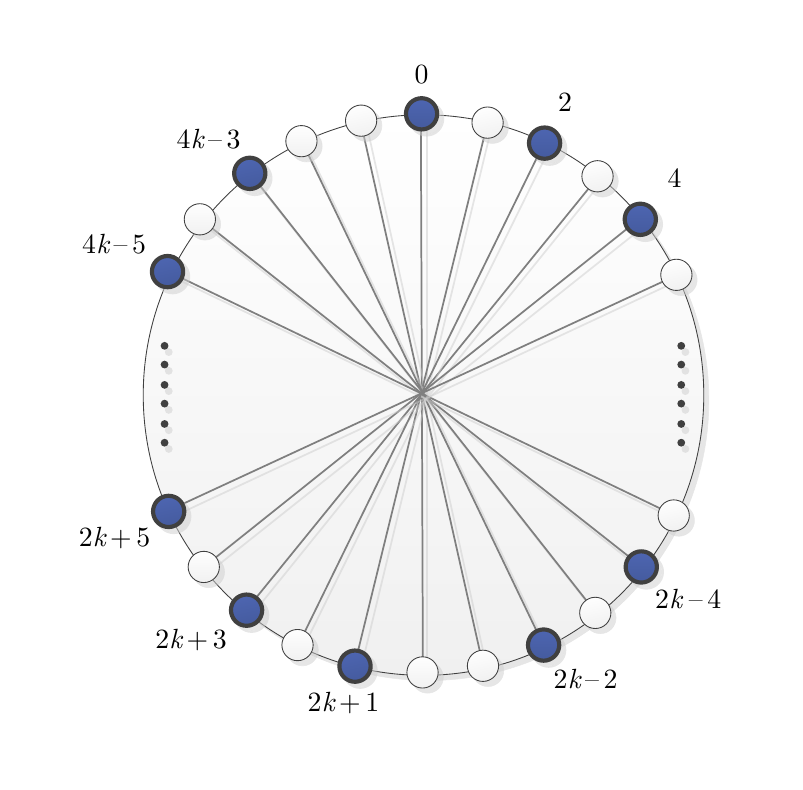}
\caption{A 3-regular cyclic Cayley graph on $4k$ vertices.  Highlighted vertices mark the set $D$ used in the proof of Theorem~\svref{thm-3-reg-cayley}.}
\svlabel{fig:3regCC}
\end{center}
\vspace{-0.5cm}
\end{figure}
Observe from Figure~\svref{fig:3regCC} that $D \decode D^+$.  Also note that $D^+$ is missing only four vertices, two on each side almost directly across from each other, and $|D| = 2k-1$.

Similar to our proof for the 5-cycle, we will prove $b_2 \ge n/2$ by listing a sequence of constraints in the LP $\mathcal{B}_2$ that sum and cancel to give us $X(\emptyset) \ge n/2$.  However, this proof differs from the 5-cycle proof because we list inequalities implied not only by the constraints in our LP but also our assumption of cyclic symmetry.  The fact that any two slices of size $i$ have the same $X$ value is used heavily in the sequence of inequalities that make up our proof.\\
First, we create $2k-1$ $X(D^+)$ terms on the right-hand-side:
\begin{align*}
(2k-2) + X(\emptyset) &\geq X(D \setminus \{0\}) & \mbox{(slope)}\\
X([1]) + X(D \setminus \{0\}) & \ge X(D^+) + X(\emptyset) & \mbox{(submod , decode)} \\
(2k-2)((2k-1) + X(\emptyset) &\geq X(D^+)) & \mbox{(slope , decode)}
\end{align*}
Now, we apply submodularity to slices of size $i = 2\ldots2k$ and an $X(D^+)$ term --- canceling all the $X(D^+)$ terms we created on the right-hand-side in the previous step.  We pick our slices so that the union is a slice missing only two vertices, and the intersection is a slice of size $i-1$.
\begin{align*}
X(D^+) + X([2k]) & \geq X([4k-2]) + X([2k-1])\\
X(D^+) + X([2k-1]) &\geq X([4k-2]) + X([2k-2])\\
&~\vdots\\
X(D^+) + X([2]) &\geq X([4k-2]) + X([1])
\end{align*}
If we sum and cancel the inequalities listed so far we have: $$2k(2k-2)+ (2k-2)X(\emptyset) + X([2k]) \ge (2k-1)X([4k-2])$$  Now, we combine all $2k-1$ of the $X([4k-2])$ terms to get full cycles.
\begin{align*}
2X([4k-2]) & \geq X(V) + X([4k-3])\\
X([4k-3]) + X([4k-2]) &\geq X(V) + X([4k-4])\\
X([4k-4]) + X([4k-2]) &\geq X(V) + X([4k-5])\\
&~\vdots&\\
X([2k+1]) + X(H[4k-2]) &\geq H(V) + H([2k])
\end{align*}
Now, we are left with:$$2k(2k-2)+ (2k-2)X(\emptyset) \ge (2k-2)X(V)$$
We can apply the constraint $X(V) \ge n$, yielding:
$$2k(2k-2)+ (2k-2)X(\emptyset) \ge (2k-2)4k$$ thus $X(\emptyset)\ge 2k$ for any feasible solution, implying $b_2 \ge 2k = n/2$.
\end{proof}

\begin{proof}[\textbf{\emph{Proof of Theorem~\svref{thm-circulant}}}]
It is easy to check that $\chibar_f$ for these graphs is $n/(k+1)$, so it is sufficient to prove that $b_2 \ge n/(k+1)$.
As we did in the proof of Theorem~\svref{thm-3-reg-cayley} we will assume that our solution $X$ has cyclic symmetry.  Suppose that $n \mod (k+1) \equiv j$.  Now, consider dividing the cycle into sections of size $k+1$ and let $S$ be the set of vertices consisting of the first $k$ in each complete section ($|S| = k(n-j)/(k+1)$).  Then by decoding $X(S) = X([-j])$ where $[-j]$ is a slice of size $n-j$.  We will also use $[j]$ to denote a set of size $j$, as in the proof of Theorem~\svref{thm-3-reg-cayley}.  Observe that if $j = 0$ then this observation alone completes the proof.

\begin{lemma}
$(k+1)X[-j] + X[k] \ge (k+1)[-j-1] + X(\emptyset)$
\svlabel{lem:gen1k_helper}
\end{lemma}

\begin{proof}
The following inequalities are true by submodularity and the cyclic symmetry of $X$.   In each inequality we apply submodularity to two slices, say of size $s$ and $t$, $s \le t$, overlapping such that their intersection is a slice of size $s-1$ and their union a slice of size $t+1$.\begin{align*}
X([-j])+X([-j]) &\ge X([-j+1]) + X([-j-1])\\
X([-j])+X([-j-1]) &\ge X([-j+1]) + X([-j-2])\\
X([-j])+X([-j-2]) &\ge X([-j+1]) + X([-j-3])\\
&~\vdots\\
X([-j])+X([-j-(k-1)]) &\ge X([-j+1]) + X([-j-k])\\
X([-j-k])+X([k]) & \ge X(\emptyset) + X([-j+1]) \qquad\qquad\mbox{(submod , decode).}
\end{align*}
Adding up all of these inequalities gives us the desired inequality.
\end{proof}

Now, if we sum together the following string of inequalities we get the bound we want on $X(\emptyset)$.  Essentially, we iteratively apply our Lemma to get us to the trivial $j=0$ case.
\begin{align*}
k(n-j) + (k+1)X(\emptyset) &\ge (k+1)X([-j]) & \mbox{(slope , decode)}\\
jk + jX(\emptyset) &\ge jX([k]) & \mbox{(slope)}\\
(k+1)X([-j]) + X([k]) & \ge (k+1)X([-j-1]) + X(\emptyset) & \mbox{(by Lemma~\svref{lem:gen1k_helper})} \\
(k+1)X([-j-1]) + X([k]) & \ge (k+1)X([-j-2]) + X(\emptyset) & \mbox{(by Lemma~\svref{lem:gen1k_helper})} \\
&~\vdots\\
(k+1)X([-1]) + X([k]) & \ge (k+1)X(V) + X(\emptyset) & \mbox{(by Lemma~\svref{lem:gen1k_helper})} \\
(k+1)X(V) &\ge kn\,.
\end{align*}
This completes the proof.
\end{proof}

\subsection{The broadcast rate of specific small graphs}
For any specific graph one can attempt to solve the second level of the LP-hierarchy directly to yield a possibly tight lower bound $\beta \geq b_2$. The following corollary lists a few examples obtained using an AMPL/CPLEX solver.
\begin{fact}\svlabel{cor-specific-graphs} The following graphs satisfy $b_2=\beta=\chibar_f$:
\begin{compactenum}
  [(1)]
  \item Petersen graph (Kneser graph on $\binom{5}{2}$ vertices): $n=10$, $\alpha = 4$ and $\beta = 5$.
  \item Gr\"{o}tzsch graph (smallest triangle-free graph with $\chi=4$): $n=11$, $\alpha=5$ and $\beta = \frac{11}2$.
\item Chvatal graph (smallest triangle-free $4$-regular graph with $\chi=4$): $n=12$, $\alpha=4$ and $\beta = 6$.
\end{compactenum}
\end{fact}


\section{Coverage functions: a proof of Lemma~\svref{lem:coverage-functions}}\svlabel{sec:ap-coverage}
Lemma~\svref{lem:coverage-functions} (\S~\svref{sec:hierarchy-proof}) will readily follow from establishing the following Lemmas~\svref{lem:coverage1} and \svref{lem:coverage2}, as it is easy to verify that the slope constraints and the i-th order submodularity constraints in our LP are equivalent to the inequalities in Eq.~\sveqref{eq:almost-LP}.



\begin{lemma}
A vector $X$, indexed over all subsets of the groundset $V$, satisfies
\begin{equation} \svlabel{eq:almost-LP}
\forall R \ne \emptyset, \forall Z \cap R = \emptyset, \;\;
\sum_{T \subseteq R} (-1)^{|R\setminus T|} X(T \cup Z) \le \left\{\begin{array}{ll}1 & \mbox{if $|R|=1$} \\
    0 & \mbox{otherwise}
\end{array}\right.
\end{equation}
if and only if it satisfies:
\begin{equation} \svlabel{eq:almost-coverage}
\forall R \subseteq \groundset, R \neq \emptyset, \;\;
\sum_{T \subseteq R} (-1)^{|T|} (X(R \setminus T) - |R\setminus T|) \le 0\,.
\end{equation}
\svlabel{lem:coverage1}
\end{lemma}

\begin{lemma}
A vector $X$, indexed over all subsets of the ground-set $V$, satisfies \sveqref{eq:almost-coverage} if and only if there exists a vector of non-negative numbers $w(T)$, defined for every non-empty vertex set $T$, such that $X(S) = |S| + \sum_{T: T \not \subseteq S} w(T)$ $\forall S \subseteq V$.
\svlabel{lem:coverage2}
\end{lemma}

\begin{proof}[Proof of Lemma~\svref{lem:coverage1}]
First, we claim that $X$ satisfies \sveqref{eq:almost-coverage} if and only if it satisfies:
\begin{equation} \svlabel{eq:middle}
\forall R \subseteq \groundset, R \neq \emptyset, \;\;
\sum_{T \subseteq R} (-1)^{|R \setminus T|} X(T) \le \left\{\begin{array}{ll}1 & \mbox{if $|R|=1$}\,, \\
    0 & \mbox{otherwise.}
\end{array}\right.
\end{equation}

Starting with the inequalities \sveqref{eq:almost-coverage}, observe that we get an equivalent set of inequalities when we switch the roles of $T$ and $R \setminus T$, as it is essentially summing over the complements of $T$ instead of $T$.    Additionally, for $|R| \ge 2$ we can remove the constant term because it is equal to the alternating sum $ \pm\sum_{i = 1}^k (-1)^k \binom{|R|}{k}k = 0$.
If $|R| = 1$ then the constant term is one.

Now, we show the equivalence of \sveqref{eq:middle} and \sveqref{eq:almost-LP}.  Clearly, if $X$ satisfies \sveqref{eq:almost-LP} then it satisfies \sveqref{eq:middle} because the inequalities in the latter are a subset of the inequalities in the former.  Now, we show by induction on the size of $Z$ that~\sveqref{eq:middle} implies \sveqref{eq:almost-LP}. Our base case, $|Z| = 0$ holds trivially.  We assume that~\sveqref{eq:middle} implies~\sveqref{eq:almost-LP} for $|Z| <|Z^*|$ and show the following inequality holds:
\[\sum_{T \subseteq R^*} (-1)^{|R^*\setminus T|} X(T \cup Z^*) \le \left\{\begin{array}{ll}1 & \mbox{if $|R|=1$} \\
    0 & \mbox{otherwise}
\end{array}\right.\tag{$\star$}\]
By our inductive hypothesis, Eq.~\sveqref{eq:middle} implies the following two inequalities from~\sveqref{eq:almost-LP}:
\begin{align*}
R &= R^* \cup \{z\}\,,\quad Z = Z^* \setminus \{z\} \tag{I}\,,\\
R &= R^*\,,\qquad Z = Z^* \setminus \{z\} \tag{II}
\end{align*} for some $z \in Z^*$.  It is easy to see that $(\star)-(\mathrm{II}) = (\mathrm{I})$, thus we can derive ($\star$) from $(\mathrm{I}),(\mathrm{II})$.
\end{proof}

\noindent\emph{Proof of Lemma~\svref{lem:coverage2}.}
Suppose there exists a vector of non-negative numbers $w(T)$, defined for every non-empty vertex set $T$, such that $X(S) = |S| + \sum_{T: T \not \subseteq S} w(T)$ $\forall S \subseteq V$ as in the statement of our Lemma.  Then rearranging, we have:

$$X(\overline{S})-|\overline{S}| = \sum_{T: T \not \subseteq \overline{S}} w(T) = \sum_{T:T \cap S \ne \emptyset} w(T)  \;\;\forall \overline{S} \subseteq V$$

Now, define $F(S) = X(\overline{S}) - |\overline{S}|$.
\begin{lemma}
The set function $F$ satisfies
\begin{equation} \svlabel{eq:F_almost-coverage}
\forall R \subseteq \groundset, R \neq \emptyset, \;\;
\sum_{T \subseteq R} (-1)^{|T|} \fcn(\overline{R} \cup T)
\leq 0.
\end{equation}
if and only if there exists a vector of non-negative numbers $w(T)$, defined for every non-empty vertex set $T$, such that
\begin{equation}
  \svlabel{eq:set-cover-fn}
F(S) = \sum_{T:T \cap S \ne \emptyset} w(T)  \;\;\forall S \subseteq V.
\end{equation}
\svlabel{lem:coverage-helper}
\end{lemma}

\begin{remark}
A set function $F$ is called a weighted set cover function if it can be written as in Eq.~\sveqref{eq:set-cover-fn}.
\end{remark}

Plugging in $X(\overline{S})-|\overline{S}|$ for $F(S)$ and noting that $\overline{\overline{R} \cup T} = R \setminus T$ it is easy to see that Lemma~\svref{lem:coverage-helper} implies our desired result.  Thus, it remains to prove Lemma~\svref{lem:coverage-helper}.

\begin{proof}[Proof of Lemma~\svref{lem:coverage-helper}]
In this proof we will be working with vectors and matrices whose rows
and columns are indexed by subsets of $\groundset$.  Let $n=|\groundset|, N = 2^n$.  Expressing $\fcn$ and $w$ as vectors with $N-1$ components
(ignoring the component corresponding to the empty set),
this equation can be written in matrix form as
\[
\fcn = M w,
\]
where $M$ is the $(N-1)$-by-$(N-1)$ matrix defined by
\[
M_{TS} = \begin{cases}1 & \mbox{if } T \cap S \neq \emptyset \\
0 & \mbox{otherwise.} \end{cases}
\]
We shall see below that $M$ is invertible.
It follows that $\fcn$ can be written as in Eq.~\sveqref{eq:set-cover-fn} if and only if $M^{-1} \fcn$ is a vector $w$ with non-negative
components.

To prove that $M$ is invertible and to obtain
a formula for the entries of the inverse matrix,
let $L$ be the $N$-by-$N$ matrix defined by
\[
L_{TS} = \begin{cases}1 & \mbox{if } T \cap S \neq \emptyset \\
0 & \mbox{otherwise.} \end{cases}
\]
In other words, $L$ is the matrix obtained from $M$
by adding a top row and a left column consisting
entirely of zeros.  Let us define another matrix
$K$ by
\[
K_{TS} = 1 - L_{TS} =
\begin{cases}1 & \mbox{if } T \cap S = \emptyset \\
0 & \mbox{otherwise.} \end{cases}
\]
We can now begin to make progress on inverting
these matrices, using the observation that both
$K$ and $K+L$ can be represented as Kronecker
products of $2$-by-$2$ matrices.  Specifically,
\begin{align*}
K &= \left( \begin{array}{rr} 1 & 1 \\ 1 & 0 \end{array}
\right)^{\otimes n} ~,\quad
 K+L = \left( \begin{array}{rr} 1 & 1 \\ 1 & 1 \end{array}
\right)^{\otimes n}.
\end{align*}
The inverse of $\left(\begin{smallmatrix} 1 & 1 \\ 1 & 0\end{smallmatrix}\right)$
is $\left(\begin{smallmatrix} 0 & 1 \\ 1 & -1\end{smallmatrix}\right)$.
We may now make use of the fact that Kronecker
products commute with matrix products, to deduce that
\begin{align*}
L \left( \begin{array}{rr} 0 & 1 \\ 1 & -1 \end{array}
\right)^{\otimes n} &=
(K+L)  \left( \begin{array}{rr} 0 & 1 \\ 1 & -1 \end{array}
\right)^{\otimes n} -
K  \left( \begin{array}{rr} 0 & 1 \\ 1 & -1 \end{array}
\right)^{\otimes n} \\
&=
\left[
\left( \begin{array}{rr} 1 & 1 \\ 1 & 1 \end{array} \right)
\left( \begin{array}{rr} 0 & 1 \\ 1 & -1 \end{array} \right)
\right]^{\otimes n} -
\left[
\left( \begin{array}{rr} 1 & 1 \\ 1 & 0 \end{array} \right)
\left( \begin{array}{rr} 0 & 1 \\ 1 & -1 \end{array} \right)
\right]^{\otimes n} \\
&=
\left( \begin{array}{rr} 1 & 0 \\ 1 & 0 \end{array} \right)^{\otimes n}
-
\left( \begin{array}{rr} 1 & 0 \\ 0 & 1 \end{array}
\right)^{\otimes n}.
\end{align*}
Examine the matrices occurring on the
left and right sides of the equation
above, and consider the submatrix obtained
by deleting the left column and top row.
On the right side, we obtain $-I$, where $I$
denotes the $(N-1)$-by-$(N-1)$ identity matrix.
On the left side we obtain $M \cdot A$, where
$A$ is the matrix obtained from
$\left(\begin{smallmatrix} 0 & 1 \\ 1 & -1\end{smallmatrix}\right)^{\otimes n}$
by deleting the left
column and top row.  This implies that
$M$ is invertible and its inverse is $-A$.
Moreover, one can verify that the
entries of $-A$ are given by
\[
-A_{TS} = \begin{cases}
0 & \mbox{if } T \cup S \neq \groundset \\
(-1)^{|T \cap S|} & \mbox{if } T \cup S = \groundset.
\end{cases}
\]

Recall that a set function $\fcn$ can be
expressed as it is in Eq.~\sveqref{eq:set-cover-fn}
if and only if $M^{-1} \fcn$ has non-negative
entries.  Now that we have derived an expression
for $M^{-1}$ we find that this criterion is
equivalent to stating that for all nonempty sets
$R \subseteq \groundset$,
\[
\sum_{S \,:\, T \cup S = \groundset}
(-1)^{|T \cap S|} \fcn(S) \leq 0.
\]
This condition is equivalent to Eq.~\sveqref{eq:almost-coverage} because every set $S$ such that $T \cup S = \groundset$
can be uniquely written as the disjoint union of
two sets $\overline{T}$ and $R = T \cap S.$
This completes the proof of Lemma~\svref{lem:coverage-helper} and subsequently proves Lemmas~\svref{lem:coverage2} and~\svref{lem:coverage-functions}.
\end{proof}

\section{Open problems}\svlabel{sec:conclusion}
\begin{compactitem}
\item  We provide an information-theoretic lower bound $b_2$ on the broadcast rate $\beta$, enabling us to answer fundamental questions about the behavior of $\beta$.  While one can have $b_2 < \beta$, what is the largest possible gap between the two parameters? Recalling that the linear program for $b_2$ contains exponentially many constraints, is there an efficient algorithm for computing $b_2$?
\item Our results include a polynomial time algorithm for determining whether $\beta=2$ for any broadcasting network. A major open problem is establishing the hardness of determining whether $\beta < C$ for a given graph $G$ and real $C>0$. While no such hardness result is known, presumably this problem is extremely difficult e.g.\ it is unclear whether it is even decidable.
\item In an effort to approximate $\beta$, we give an efficient multiplicative $o(n)$-approximation algorithm for the general broadcasting problem. Can we obtain an approximation of $\beta$ (even for case of undirected graphs) within a multiplicative constant of $n^{1-\epsilon}$ for some fixed $\epsilon > 0$?
\item Using certain projective-Hadamard graphs introduced by Erd\H{o}s and R\'enyi, we show that the broadcast rate can be uniformly bounded while its upper bound $b_n$ is polynomially large. Is the scalar capacity $\beta_1$ of these graphs unbounded as the field characteristic $q$ tends to $\infty$?



\end{compactitem}

\bigskip
\noindent \textbf{Acknowledgment.} We thank Noga Alon for useful discussions.


\end{document}